\PassOptionsToPackage{table}{xcolor}
\documentclass[runningheads]{llncs}
\usepackage[T1]{fontenc}
\usepackage{graphicx}

\newcommand{\red}[1]{{\color{red} #1}}

\newcommand{\myparagraph}[1]{{\noindent \bf #1}}
\usepackage{color}
\usepackage{amsmath}
\usepackage{amsfonts}
\usepackage{mathtools}
\usepackage{bbm}
\usepackage{algorithm}
\usepackage{algpseudocode}
\usepackage{booktabs}      %
\usepackage{pgfplotstable} %
\usetikzlibrary{arrows.meta}
\usepackage{pdflscape}     %
\usepackage{siunitx}       %
\pgfplotsset{compat=1.18}
\algrenewcommand\algorithmicrequire{\textbf{Input:}}
\algrenewcommand\algorithmicensure{\textbf{Output:}}

\newcommand{\sk}{\ensuremath{\mathtt{sk}}}

\newcommand{\ppt}{\mathsf{PPT}}

\newcommand{\setup}{\mathsf{Setup}}

\newcommand{\NN}{\mathbb{N}}

\newcommand{\vk}{\mathsf{vk}}

\newcommand{\adv}{\mathcal{A}}

\newcommand{\prover}{\mathcal{P}}
\newcommand{\verifier}{\mathcal{V}}

\newcommand{\ignore}[1]{}

\newcommand{\simu}{\mathsf{Sim}}

\newcommand{\NP}{NP}

\newcommand{\negligible}{\nu}

\newcommand{\remove}[1]{}

\newcounter{itemcount4}

\newcommand{\crs}{\mathtt{CRS}}

\newcommand{\relation}[1]{\mathcal{R}_{#1}}

\newcommand{\sample}{\overset{\$}{\gets}}

\let\llncsproof\proof
\let\llncsendproof\endproof
\let\proof\relax
\let\endproof\relax
\usepackage{booktabs}
\usepackage{enumitem}
\usepackage{latexsym}
\usepackage{mathtools}
\usepackage{multirow}
\usepackage{multicol}
\usepackage{mwe}
\usepackage{makecell}
\usepackage{pgfplots}
\usepackage{pifont}
\usepackage{ragged2e}
\usepackage{rotating}
\usepackage{subcaption}
\usepackage{tabu}
\usepackage{tabularx}
\usepackage{threeparttable}
\usepackage{tikz}
\usepackage[normalem]{ulem}
\usepackage{url}
\usepackage{verbatim}
\usepackage{wrapfig}
\usepackage{xcolor}
\usepackage{xspace}
\usepackage[framemethod=tikz]{mdframed}
\usepackage{algpseudocode}
\usepackage[most]{tcolorbox}

\usepackage{amsfonts,amsmath,amssymb,amsthm}
\usepackage{multirow}
\usepackage{hyperref}
\usepackage{svg}
\usepackage{float}
\usepackage[
    n,
    operators,
    advantage,
    sets,
    adversary,
    landau,
    probability,
    notions,    
    logic,
    ff,
    mm,
    primitives,
    events,
    complexity,
    asymptotics,
    keys]{cryptocode}
\usepackage{algorithm}
\algrenewcommand\algorithmicrequire{\textbf{Input:}}
\algrenewcommand\algorithmicensure{\textbf{Output:}}

\usepackage{breakcites}

\expandafter\def\csname ver@etex.sty\endcsname{}

\makeatletter
\def\th@plain{%
  \thm@notefont{}%
  \itshape %
}
\def\th@definition{%
  \thm@notefont{}%
  \normalfont %
}
\makeatother

\usepackage{mdframed}
\newmdenv[%
  backgroundcolor=gray!10,
  linecolor=black,
  linewidth=1pt,
  nobreak=true,
  innerleftmargin=5pt,
  innerrightmargin=10pt,
  innertopmargin=1pt,
  innerbottommargin=10pt,
  font=\normalfont,
]{theorembox}

\let\proof\llncsproof
\let\endproof\llncsendproof
\input{symbols_env}

\newtheorem{assumption}{Assumption}
\usepackage{booktabs}
\usepackage[table]{xcolor}
\usepackage{tabularx}
\definecolor{SplitA}{HTML}{E3F2FD} %
\definecolor{SplitB}{HTML}{FFF3E0} %
\definecolor{SplitC}{HTML}{E8F5E9} %

\providecommand{\Description}[2][]{ }

\begin{document}

\title{Proof-of-Machine-Learning-Inference: \\
Mining Knowledge, Not Hashes}
\titlerunning{PoML: Mining Knowledge, Not Hashes}
\newif\ifanonymous
\anonymousfalse
\ifanonymous
\author{}
\authorrunning{}
\institute{}
\else

\author{Wan Ki Wong\inst{1} \and Michele Ciampi\inst{1} \and Rik Sarkar\inst{1}}
\authorrunning{W. K. Wong et al.}
\institute{University of Edinburgh, Edinburgh, United Kingdom\\
\email{\{thomas.wong,michele.ciampi,rik.sarkar\}@ed.ac.uk}}
\fi

\maketitle

\begin{abstract}
Bitcoin-style Proof-of-Work (PoW) secures permissionless \\blockchains but spends substantial computation on hash puzzles which have no utility other than extending the chain. 
At the same time, we have seen that the rise of generative AI has created significant demand for computational resources that are unfortunately available only to a few players. We propose a solution that solves both problems at the same time. In particular, we present Proof-of-ML-Inference (PoML), a consensus protocol in which miners perform user-requested machine-learning inference as part of their consensus algorithm. 
We identify exactly what computational assumptions are needed for ML inference to safely replace hash work, formalize those assumptions, 
and give evidence that LLMs and diffusion models satisfy them.
We then formally analyse the security of our protocol by reducing the security our protocol to the security of the  Bitcoin Backbone model. Finally, we empirically evaluate block-time stability, showing that, in a simulated environment, PoML maintains a block production rate close to its calibrated block-time target.

\end{abstract}

\section{Introduction}
\label{sec:introduction}
\renewcommand{\red}[1]{#1}

Public generative models with open weights are powerful tools for democratizing the use of AI, with model hubs such as Hugging Face~\cite{hf_state_of_os_spring26} supporting a large ecosystem of publicly available models for text analysis, image processing, scientific research and more~\cite{longpre2025economies}. Public models also support independent auditing and public evaluation of model safeguards in response to growing AI-safety concerns~\cite{zhong2026watch,internationalAISafetyReport2026}. However, the execution of these models is still expensive~\cite{luccioni2024power} and is impractical for a small or medium sized organisation. A natural approach to crowdsource and govern a shared model execution system is to use blockchains, where users can freely submit large batches of queries and model executors can be paid for their contributions in answering these queries.

Proof-of-work-based blockchains are particularly suitable for a democratic machine learning platform as they enable truly permissionless environments,  which allow anyone to participate. One well-known drawback of traditional proof-of-work systems, such as Bitcoin, is that they are computationally wasteful, as they consume immense amounts of energy to run a hash mining protocol for consensus~\cite{polytechnique2024bitcoin,de2022revisiting}. In this paper, we propose to instead use the generative AI computation itself as proof-of-work.

We call our proof-of-work Proof-of-Machine-Learning (PoML) and design a blockchain system based on it. PoML rechannels the computation required to solve standard PoW puzzles into computation useful for answering inference queries. This aims to make large-scale AI compute available more broadly (to anyone willing to pay a fee). Designing a PoML blockchain protocol that provides a similar level of security as Bitcoin, while replacing standard PoW with inference queries computed using current generative AI technologies, is the major challenge faced in this paper. Below we summarize our contributions in more detail.

\subsection{Our Contributions}

The high level description of our system is quite simple. Users submit their queries $q$ to a common pool. Miners, acting as model executors, pick up queries at random and compute the corresponding answers $y$ using a public model. A miner successfully generates a block if it finds a solution $y$ such that $H(y||h)<T$ for a target value $T$, where $h$ is the hash of the latest block. The more inference queries a miner answers, the more likely it is to generate a new block and receive the associated rewards.

Making this approach work faces various challenges. 
The first is that the adversary could fabricate and submit queries queries that are all \emph{easy} to solve. For example, the adversary could issue a set of malicious queries that can be easily solved based on some pre-computation, or if part of the computation done to answer a certain query can be cached and re-used. These attacks give a leverage to the adversary that honest parties would not have, as they are supposed to select queries randomly and execute the model using the standard procedure.

To solve this problem, we consider models that take randomness as an additional input, and guarantee that the only way to compute an inference query (even on a maliciously chosen query) is via a circuit of depth at least $d$.  Hence, this property guarantees that, as long as the auxiliary random string is uniformly randomly sampled, there is a certain amount of computation that a node needs to do (i.e., there are no shortcuts that an adversary can take to compute an output of a model).
We formally define this special property that a model should have, calling it \emph{Computational Independence (CI)}. As one would expect, in our application we use such a model using randomness taken from the blockchain, and assuming the computational independence of the model, we reduce the security of our protocol to the security of the Bitcoin Backbone~\cite{GKL20}.

The other main result of the paper is to show that existing models do satisfy the property of CI. In particular, we formally argue that a standard class of stochastic diffusion models naturally satisfies the notion. Hence, our protocol can be instantiated to, for example, answer queries related to image generation \cite{rombach2022high}, video editing \cite{ho2022video}, and audio synthesis \cite{kong2021diffwave}.

We also argue, although only empirically, that under certain conditions large language models (LLM) also seem to satisfy CI. We show that on the same inputs (but different randomness), LLMs produce different outputs with a probability rapidly increasing with length. This suggests that, even when computing inference of similar queries, the traces of the computation differ significantly for large outputs. It remains an interesting open question to prove that there are LLM which are also CI.
We finally implement and evaluate our protocol, showing that PoML can achieve stable block production rates with lower observed block-time variance than the PoW baseline.

\noindent{\bf Additional subtleties to be solved. }
At the beginning of the section we said that to win a lottery it is sufficient to find a solution $y$ such that $H(y)<T$. This would clearly waste a lot of work, as the miner would publish only $y$, discarding all the solutions to the other queries used to run the lottery. A simple solution to this problem would be to change the winning condition to be $H(y_1,\dots,y_n)<T$ (where $y_i$ is a solution to a query $q_i$). This would require a lottery winner to disclose $n$ solutions. Unfortunately this does not work, as an adversary could query $H$ using any subset of the solutions found so far, and increase his chance of winning the lottery. Making this the honest strategy also would not work, as it could be that $H(y_n)<T$, but $H(y_1,\dots,y_n)\geq T$. Hence, only the solution for $y_n$ would be disclosed. To solve this problem we instead ask the miner to use as the randomness needed to solve the $j$-th query, the hash of the previous solution found. This create a sub-chain of solutions, and the miner wins if the head $h$ of such a chain is that $H(h)<T$. Changing the position of an element in the chain, would require recomputing all the inference queries whose randomness depend on it.
An unrelated problem, is that only one miner can win the lottery for each block and the work performed by the other miners would be lost (recall that the inference query is implicitly tied to a specific block and as such cannot be \emph{recycled}). We mitigate this problem by allowing miners that did not win the block-lottery to submit the inference queries they solved in the form of a standard transaction, and still get a (smaller) reward. This reward is issued only for inference queries for which a solution has not been posted yet, and only for solutions related to queries that are not too old. While this method does not fully solve the problem, as multiple miners could be working on exactly the same queries, we show that the the wasted work is minimal: for example, where number of pending query is 100x of the number of miners, the fraction of deuplicated work remains below 5$\%$ in our experiments.  

Finally, in the sketched protocol the solutions to the queries are made public. We note that this is not preferable, as a user paying to get a certain query to be solved, may not want other users to have free access to that solution. We solve this problem by encrypting the solutions with the public-key of the user that issued the related query. A zero-knowledge proof will be send along with the encryption, to prove that the encryption contains a correct answer to a given query.

\subsection{Related works}

Several proposals replace energy-intensive Proof-of-Work hashing with Proof of Useful Work (PoUW) including machine learning.  Proof of Learning (PoLe)~\cite{pole} uses training of neural networks as the useful work. The protocol binds training to previous blocks to prevent precomputation, but at the cost of requiring a modified model architecture. Other ML-training based methods aim for different objectives. Lihu et al.~\cite{lihu2020pouw} limit nonce grinding by deriving nonces from training byproducts such as gradient updates. Proof of Training (PoT)~\cite{pot} combines elements of Proof-of-Stake and Proof-of-Activity. Proof-of-Collaborative-Learning (PoCL)~\cite{pocl}  focuses on federated learning with peer-evaluations and reward the top-ranked contributors. Various other proof of useful work proposals use heavy computational tasks such as Matrix multiplication~\cite{cryptoeprint:2025/685}, DNA alignment~\cite{ileri2016coinami} and finding prime numbers~\cite{king2013primecoin}. While these tasks are of high value, their demand is restricted to scientists and engineers. In contrast, generative AI inference used in PoML has wide applicability among non-specialists and it works with standard types of generative models, making the protocol more widely useful.

Formal security analyses of PoUW are provided in~\cite{FKPR22} in the context of solving combinatorial optimization problems. Like PoML, it adopts the Bitcoin Backbone's flat computation model and proves chain growth, chain quality, and common prefix. To prevent easy-instance selection and pre-computation, Ofelimos places the useful work between a pre-hash and a post-hash, with the pre-hash serving as a traditional PoW puzzle. A formal ledger-property proof for a deep-learning-based PoUW is shown in D-PoDL,~\cite{SLT23}, using a similar ``hash-training-hash'' structure. Both protocols retain wasteful pre-hash nonce grinding, while PoML avoids this by feeding a block-derived seed into a Verifiable Random Function (VRF) and hashing encrypted inference output directly for the mining lottery.

Finally, Oleksak et al.~\cite{pouw} propose a consensus-layer marketplace for general-purpose SNARK generation rather than an ML task. They bind each proof to the block being produced by injecting an integrity parameter into the arithmetic circuit, and chain successive proofs to prevent proof theft and proof replay. Their main useful work is the proof generation while PoML's is machine learning inference, with ZKPs used only for verification.

\section{Preliminaries}
\label{sec:preliminaries}

We assume familiarity with standard cryptographic notions,  including Non-interactive zero-knowledge proofs, Public-Key encryption scheme, Verifiable Random Functions (VRFs) and Digital Signatures. Formal definitions are provided in Appendix~\ref{app:zk}--\ref{app:sig} for reference. Below we provide short introduction to blockchain and generative models.

\subsection{Blockchain and Proof-of-Work Consensus}

A \emph{blockchain} is a decentralized, append-only ledger of ordered blocks, each linking to its predecessor by hash and carrying transactions~\cite{nakamoto2008bitcoin,bonneau2015sok}. In Proof-of-Work (PoW), miners search for a nonce $\mathit{ctr}$ such that $H(\mathit{ctr},G(s,\mathsf{tx}))<T$, where $H,G$ are hash functions, $s$ is the predecessor hash, and $\mathsf{tx}$ is the block content. Modeling $H$ as a random oracle gives independent success probability $p=T/2^\kappa$ per query, so the expected work is $1/p$. $T$ is adjusted dynamically to maintain the target block rate.

For security analysis, we use the Bitcoin backbone model of Garay, Kiayias, and Leonardos~\cite{GKL20}, which bounds each party to $q$ random-oracle queries per round and assumes a strict honest computational majority. We recall that the three fundamental properties that a secure blockchain system should satisfy are  common prefix (CP), chain quality (CQ), and chain growth (CG); the formal formulations can be found in~\cite{GKL20}.

\subsection{Generative models}
\label{subsec:ml-background}

\subsubsection{Diffusion Models}
\label{subsec:diffusion-models}

\emph{Diffusion models}~\cite{ho2020ddpm} generate outputs such as images by repeatedly denoising high dimensional Gaussian noise. Starting from $x_T\sim\mathcal{N}(0,I_d)$, a learned denoiser $f_\theta$ -- typically a U-Net structured neural network -- sequentially produces $x_T,\ldots$ upto $x_0$ -- the final model output. Text prompt based models use an embedding of the prompt to steer the denoising process toward relevant outputs. 

Many widely used diffusion models are stochastic models~\cite{ho2020ddpm} that add Gaussian noise at each step of denoising to improve generation quality and diversity. In notation: $x_{t-1}=f_\theta(x_t,x_\mathsf{u},t)+\sigma_tz_t$ for $t=T,T{-}1,\ldots,2$, with fresh $z_t\sim\mathcal{N}(0,I_d)$. This per-step noise makes each denoising trajectory unique with high probability and is the basis of our computational independence argument in Section~\ref{sec:our-model}.

\subsubsection{Large Language Models}
\label{subsec:large-language-models}

An \emph{autoregressive large language model (LLM)}, typically a causal Transformer~\cite{vaswani2017attention} takes as input a string of tokens and output the next token. Given prompt tokens $x=(x_1,\ldots,x_{t-1})$, the main function of an LLM is to generate a probability distribution $p_{\theta,t}(\cdot\mid x,y_{<t})$ for the next token. A random sample from this distribution is the output, which is appended to $x$ for the next iteration. 
Thus a $k$-token response $y$ has probability $\Pr[Y=y\mid x]=\prod_{t=1}^{k}p_{\theta,t}(y_t\mid x,y_{<t})$.

The essential component of an LLM is an {\em attention block} that determines via linear projections, the relation between token pairs in $x$. Several attention blocks in parallel make a {\em transformer}, and several transformers in sequence form the main pipeline of an LLM producing the proabbility distribution. We investigate compatibility of LLMs using GPT family models~\cite{radford2019language,brown2020language,openai2023gpt4}.

\label{sec:problem}

\section{Computational independence and PoML-compatible models}
\label{subsec:model-class}

In our protocol, users submit model queries to a mempool, miners retrieve these requests and compete to produce a valid block by including outputs of model inference and generating zero-knowledge proofs about the fact that the inference has been computed correctly (we refer to App.~\ref{app:overview} for a schematic representation). We define here a class of ML models suitable for use in PoML and formally define the properties that models of such a class must satisfy.

Let us write stochastic ML models $M_\theta$ parameterised by weights $\theta$ that take a \emph{user input} $x$ (e.g., a text prompt) and an ordered randomness collection $\mathcal{R}=(z_1,\ldots,z_{T(x)})$. The public model specification fixes $T(x)$ and the interpretation of each $z_t$, and $\mathcal{R}$ contains all inference randomness. The model outputs $y=M_\theta(x,\mathcal{R})$ in a variable-length output space and is deterministic once $(x,\mathcal{R})$ is fixed.

For each $(x,\mathcal{R})$, let $\mathsf{Depth}_\theta(x,\mathcal{R})$ be the depth\footnote{Here, we use the standard notion of circuit depth: the length of the longest sequence of dependent gates, independently of any ML-specific terminology.} of the circuit executed by an honest evaluation. Informally our notion of \emph{Computational independence} requires that no adversary can compute a model inference output unless it has access to a circuit of a certain depth, regardless of the computation they have performed before knowing the randomness to be used.

\begin{definition}[$\varepsilon_i$-computational independence]
\label{def:delta-reuse}
Let $M_\theta$ be a model, $\varepsilon_i \in [0,1)$, and $\mathcal{A}=\{(\mathcal{A}_1,\mathcal{A}_2)_\kappa\}_{\kappa\in\mathbb{N}}$ a probabilistic polynomial-time adversary where $\mathcal{A}_1$ outputs arbitrary state $\mathsf{st}$. We say $M_\theta$ is \emph{$\varepsilon_i$-computationally independent} if, for every such $\mathcal{A}$ satisfying $\mathsf{Depth}_{\mathcal{A}_2}(\mathsf{st},x,\mathcal{R})\leq\allowbreak(1-\varepsilon_i)\,\mathsf{Depth}_\theta(x,\mathcal{R})$ for every $(x,\mathcal{R})$, there exists a negligible function $\mu(\cdot)$ such that for all $\kappa \in \mathbb{N}$:
\begin{equation}
    \Pr\!\left[\, y^* = M_\theta(x^*, \mathcal{R}^*) \;:\;
    \begin{array}{l}
        (x^*,\mathsf{st}) \leftarrow \mathcal{A}_1(M_\theta) \\[2pt]
         \mathcal{R}^* \leftarrow (\{0,1\}^{\kappa})^{T(x^*)} \\[2pt]
        y^* \leftarrow \mathcal{A}_2(\mathsf{st},x^*, \mathcal{R}^*)
    \end{array}
    \right] \;\leq\; \mu(\kappa).
    \label{eq:online-depth-bound}
\end{equation}
\end{definition}

\begin{definition}[PoML-compatible model]
\label{def:model}
$M_\theta$ is said to be PoML-compatible if $M_\theta$ is $\varepsilon_i$-computationally independent for
some fixed $\varepsilon_i$ in the PoML protocol. We denote $\mathcal{K}$ as
the class of all PoML-compatible models.
\end{definition}

Intuitively, PoML requires that prior work on any other inputs, including the
same user input $x$ with a different randomness
$\mathcal{R}'\neq\mathcal{R}$, does not yield a shortcut for computing the
target input $(x,\mathcal{R})$ in substantially less time.  The computational
independence property formalizes this requirement by forcing an online
adversary to incur nearly the same cost as an honest work (with no saved states) of the same evaluation, even after 
polynomial-time precomputation.

\subsection{PoML-Compatible Generative Models}
\label{sec:our-model}

Both diffusion models and LLMs interleave a  deterministic operation and a random one to generate samples of relaistic outputs. We argue that this structure makes them both compatible with PoML.

\subsubsection{Stochastic Diffusion Models.} The randomness at the initialization of diffusion models as well as noise at every step makes computational independence highly probable. The full analysis can be seen in Appendix~\ref{app:cia}. The analysis makes use of the fact that the space of operation of diffusion models (the latent space) is large. In this high dimensional space, the probability is small that two different generation processes $u$ and $v$ will ever coincide. However, if we assume that at some step processes $u$ and $v$ are both at location $x_t$, the two processes will add different high dimensional Gaussian noises $z_u$ and $z_v$ to obtain the next step. For a space of dimension $d$, the probability that $z_u$ and $z_v$ are within distance $\delta$ of each other is $\approx \delta^d$ for small values of $\delta$. So, for example, if $\delta$ is the quantization resolution of the cryptographic system, then with high probability (with increasing dimension), the two processes diverge. Theorem~\ref{thm:cia-theorem} shows the analysis for the $T$ step denoising process. 

The other concern here is that well separated processes may produce parallel denoising vector outputs and thus have some redundant computations within the neural network $f_\theta$. We argue that this is unlikely, given that these vectors need to have different directions in high dimensional space to approach the correct submanifold of realistic images. In experimental tests, two inputs that differ only by a small gaussian noise with standard deviation $0.001$, the minimal difference between any intermediate vector is $2.93\times10^{-3}$, hence no meaningful computational states can be cached and reused. See Appendix~\ref{app:cia} and Theorem~\ref{thm:amort-diff} for further details.

\subsubsection{Large Language Models.} In LLMs, the deterministic step consists of the transformer sequence, while the randomized step is the sampling form the distribution found by the transformers. This setup is somewhat different from diffusion models, as unlike diffusion models, LLMs do not have a random starting point, which makes the initial computations highly predictable. Since the distribution at the sampling step is not gaussian and not in our control. A concentrated distribution can cause different executions to produce the same next token. 

We address these by adding noise to the starting point. LLMs operate by embedding each token in the input into a high dimensional space and thus converting them into numerical inputs. To introduce randomness at the start, we can add a small gaussian noise to each embedding. This technique is used in RESTA~\cite{hase2025smoothedembeddingsrobustlanguage} and Soft Reasoning~\cite{pmlr-v267-zhu25ae} to improve the robustness of the model.

Once the input to the transformer is perturbed, all the intermediate activations, until the final output logits, will all be a little different. We validated this via experiment on some benchmark datasets on the GPT-2 model, as shown in figures \ref{fig:llm-layer-linf} and \ref{fig:llm-prefix-collisions} in App.~\ref{app:llm-perturbation}. 
Figure~\ref{fig:llm-layer-linf} shows between the inference of two exact same prompts that adding noise to the model leads to numerical differences of at least $2.03\times10^{-3}$ . Figure~\ref{fig:llm-prefix-collisions} further shows that the decoded output tokens differ after 5 tokens with >99\% probability. This difference is enough in PoML, because direct reuse of any earlier computation is prevented, and an adversary will need to redo the computation in order to generate the zero-knowledge proof. We formalize this as Assumption~\ref{ass:llm-no-reuse} in Appendix~\ref{app:llm-perturbation}.

\section{The PoML Protocol}
\label{sec:protocol}

This section specifies the Proof-of-ML-Inference (PoML) consensus protocol for any fixed, public, PoML-compatible model $M_\theta$. It uses the following cryptographic primitives:
\begin{itemize}[itemsep=0em,nosep]
  \item hash functions $G,H:\{0,1\}^{*}\to\{0,1\}^{\kappa}$, modelled as random oracles for block and lottery bindings;
  \item a ZKP-friendly collision-resistant hash $H_{\mathsf{zk}}:\{0,1\}^{*}\to\{0,1\}^{\kappa}$ for binding user queries;
  \item a public-key encryption scheme $(\mathsf{KeyGen}_{\mathsf{pke}},\mathsf{Enc},\mathsf{Dec})$ satisfying Definition~\ref{def:pke-app};
  \item an unbiasable VRF $(\mathsf{VRF.Gen},\mathsf{VRF.Eval},\mathsf{VRF.Vfy})$ satisfying Definition~\ref{def:vrf-app};
  \item a digital signature scheme $(\mathsf{KeyGen}_{\mathsf{sig}},\mathsf{Sign},\mathsf{Verify}_{\mathsf{sig}})$ satisfying Definition~\ref{def:sig-app}; and
  \item an NIZK proof-of-knowledge system $(\mathsf{Setup}_{\mathsf{zkp}},\prover,\verifier)$ satisfying Definition~\ref{def: NIZK systems} for the NP-language $\mathcal{L}_{\mathsf{PoML}}$.
\end{itemize}
 \[
\mathcal{L}_{\mathsf{PoML}}
=\left\{
\begin{array}{@{}l@{}}
 (\mathsf{qid},h,\theta,\mathcal{R},\mathsf{ct},\omega,\mathsf{pk}_u,
  \mathsf{vk}^{\mathsf{VRF}}_{\mathsf{enc}},r,
  \pi^{\mathsf{VRF}}_{\mathsf{enc}}):\\[0.2em]
 \exists\,(x,y,r_{\mathsf{enc}})\ \text{such that}\\
 M_\theta(x,\mathcal{R}) = y,\\
 H_{\mathsf{zk}}(x\parallel\mathsf{qid})=h,\\
 C_{\theta,\Pi}(x,\mathcal{R}) = \omega, \\
 \mathsf{VRF.Vfy}(\mathsf{vk}^{\mathsf{VRF}}_{\mathsf{enc}},
        r\parallel\mathsf{qid},r_{\mathsf{enc}},
        \pi^{\mathsf{VRF}}_{\mathsf{enc}})=1,\\
 \mathsf{ct}=\mathsf{Enc}(\mathsf{pk}_u,
        y\parallel\mathsf{qid};r_{\mathsf{enc}}).
\end{array}
\right\}.
\]

The model $M_\theta$ accepts a user input $x$ and an ordered collection $\mathcal{R}$ of inference randomness, with $\mathcal{R}=(z_1,\ldots,z_T)$, and produces deterministic output $M_\theta(x,\mathcal{R})$. The protocol derives $\mathcal{R}$ from a fresh base seed $r$: for each $t\in[T]$, the inference VRF on $r\parallel t$ yields $z_t$ and a proof of correct evaluation. Thus $\mathcal{L}_{\mathsf{PoML}}$ asserts correct use of $\mathcal{R}$ and encryption of the output under the user's public key.

We use two special transaction types, introduced below:
\begin{itemize}[itemsep=0em,nosep]
  \item $\mathsf{MINER\_REGISTRATION}$ binds a miner's three public keys; and
  \item $\mathsf{RESPONSE\_QUERY}$ carries a completed response for a query not included in a winning proof chain.
\end{itemize}

\paragraph{Query complexity.}
Since the output length is variable, different queries may require different computational effort. We assume a public, easily computable integer-valued complexity function $C_{\theta,\Pi}(x,\mathcal{R})\geq1$. It accounts for model inference, ZKP generation, and auxiliary work such as randomness generation and encryption. This complexity function determines both the fee a user need to pay for the query. 

\subsection{Initialization}
\label{sec:poml-system-def}

Let $\mathcal{M}=\{1,\ldots,n\}$ denote the active miners. Each miner $m\in\mathcal{M}$ holds an identical model $M_{\theta}\in\mathcal{K}$ and three independent key pairs: identity $(\mathsf{vk}^{\mathsf{sig}}_m,\mathsf{sk}^{\mathsf{sig}}_m)$, inference VRF $(\mathsf{vk}^{\mathsf{VRF}}_{\mathsf{inf},m},\mathsf{sk}^{\mathsf{VRF}}_{\mathsf{inf},m})$, and encryption VRF $(\mathsf{vk}^{\mathsf{VRF}}_{\mathsf{enc},m},\mathsf{sk}^{\mathsf{VRF}}_{\mathsf{enc},m})$.

Write $\mathsf{VK}_m=(\mathsf{vk}^{\mathsf{sig}}_m,\mathsf{vk}^{\mathsf{VRF}}_{\mathsf{inf},m},\mathsf{vk}^{\mathsf{VRF}}_{\mathsf{enc},m})$ for the miner's public-key triple. A $\mathsf{MINER\_REGISTRATION}$ transaction is $T_{\mathsf{MINER\_REGISTRATION}}=(\mathsf{VK}_m,\sigma_m)$. It carries $\sigma_m\leftarrow\mathsf{Sign}(\mathsf{sk}^{\mathsf{sig}}_m,\mathsf{VK}_m)$. The registration is valid only if the signature verifies, $\mathsf{Verify}_{\mathsf{sig}}(\mathsf{vk}^{\mathsf{sig}}_m,\mathsf{VK}_m,\sigma_m)=1$, all three keys are well formed, and none is already registered; it becomes active after the containing block is accepted.

Each user $u$ holds signing keys $(\mathsf{vk}^{\mathsf{sig}}_u,\mathsf{sk}^{\mathsf{sig}}_u)$ and encryption keys $(\mathsf{pk}_u,\mathsf{sk}^{\mathsf{pke}}_u)$. The CRS for the NIZK system is fixed before the protocol starts. The genesis block binds the model, its complexity and termination specification, fee parameters, and initial miner set: $B_0=(\mathsf{CRS},h_{\theta},\mathcal{PK}^{(0)})$, where $h_{\theta}=H_{\mathsf{zk}}(\theta)$ and $\mathcal{PK}^{(0)}$ is the initial set of valid $\mathsf{VK}_m$ values.

The blockchain also defines a public difficulty parameter $0<D\leq 2^{\kappa}$. For integer complexity $C\geq1$, define $\tau_D(C)=1-\left(1-D/2^\kappa\right)^C$ and $D_C=\left\lfloor 2^\kappa\tau_D(C)\right\rfloor$. Here $\tau_D(C)$ is the ideal winning probability for a candidate of combined complexity $C$, and $D_C$ is the corresponding hash threshold.

\subsection{Fee Structure}
For a query included at height $e$, the fixed per-unit fee is $\phi_e=\phi_{\mathsf{burn},e}+\phi_{\mathsf{solve},e}+\phi_{\mathsf{include},e}$, where the burn fee is burned, the solve fee goes to the solving miner, and the include fee to the block producer (which may differ; see Sec.~\ref{sec:response-transactions}). The total fee for a query is $\phi_e\,C_{\theta,\Pi}(x,\mathcal{R})$. The user supplies a maximum fee $F_Q^{\max}$, which caps the maximum deducted fee; a response is given only when the total work costs at most this amount.

\subsection{Query Submission}
\label{sec:query-registration}

We use the synchronous diffusion-network model of the Bitcoin Backbone protocol~\cite{GKL20}: in each round, any party may \emph{diffuse} a message received by all honest parties by that round's end. A query transaction $T_{\text{query}}$ diffused in round $r$ enters every honest miner's local \emph{mempool} by round $r+1$, from which miners select pending queries when constructing a block. The raw query $x_u$ is made available off-chain, for example through a content-addressed service such as IPFS; availability and retrieval are covered by the sufficient-query-availability assumption in Appendix~\ref{app:reduction}.

A query is $Q=(\mathsf{qid},\allowbreak\mathsf{vk}^{\mathsf{sig}}_u,\allowbreak\mathsf{pk}_u,\allowbreak\mathsf{taskID},h_u,h_{\mathsf{exp}},F_Q^{\max},\sigma_u)$, where $\mathsf{qid}=H(\mathsf{vk}^{\mathsf{sig}}_u\parallel\mathsf{taskID})$, $h_u=H_{\mathsf{zk}}(x_u\parallel\mathsf{qid})$, and $\sigma_u\leftarrow\mathsf{Sign}(\mathsf{sk}^{\mathsf{sig}}_u,Q\setminus\{\sigma_u\})$. Here $x_u$ is the raw query transmitted off-chain, $h_{\mathsf{exp}}$ is the block height the user last expects a response, and $F_Q^{\max}$ is the maximum fee. The user sets $\mathsf{taskID}$ to its number of submitted queries (1st=1, 2nd=2, etc.) so that $\mathsf{qid}$ is unique.

Relative to chain state $\mathcal S$ and candidate height $h$, a query is
\emph{eligible} to mine if and only if:
\begin{enumerate}[nosep]
    \item the $\mathsf{taskID}$ is unique amongst all queries the user submitted; 
  \item its fields and public keys are well formed,
  $\mathsf{qid}=H(\mathsf{vk}^{\mathsf{sig}}_u\parallel\mathsf{taskID})$
  and $h_u=H_{\mathsf{zk}}(x_u\parallel\mathsf{qid})$ hold, and
  $\mathsf{Verify}_{\mathsf{sig}}(\mathsf{vk}^{\mathsf{sig}}_u,
    (Q\setminus\{\sigma_u\}),\sigma_u)=1$;
  \item its $\mathsf{qid}$ is not completed in $\mathcal S$ and
  $h\leq h_{\mathsf{exp}}$;
  \item the account identified by
  $\mathsf{vk}^{\mathsf{sig}}_u$ has at least $F_Q^{\max}$ available.
\end{enumerate}

\subsection{Consensus Algorithm}
\label{sec:consensus-alg}

A block produced by miner $m$ is $B=(s,\mathsf{tx},\mathsf{VK}_m,\mathcal{X},\mathcal{Y},\Pi,\Pi^{\mathsf{VRF}})$, where $s=H(B_{\mathsf{prev}})$ and $\mathsf{tx}$ is a frozen list of standard coin, $\mathsf{MINER\_REGISTRATION}$, and $\mathsf{RESPONSE\_QUERY}$ transactions. $\mathcal{X} = (Q_1,\ldots,Q_k)$ is the ordered tuple of included query transactions; $\mathcal{Y} = (\mathsf{ct}_1, \ldots, \mathsf{ct}_k)$ is the corresponding tuple of encrypted inference outputs; $\Pi = (\pi_1, \ldots, \pi_k)$ is the tuple of ZK proofs; and $\Pi^{\mathsf{VRF}} = (\Pi^{\mathsf{VRF}}_1, \ldots, \Pi^{\mathsf{VRF}}_k)$ is the tuple of per-query VRF transcripts, where each $\Pi^{\mathsf{VRF}}_i = \bigl((z_{i,t}, \pi^{\mathsf{VRF}}_{i,t})_{t=1}^{T(x_i)},\, \pi^{\mathsf{VRF}}_{\mathsf{enc},i}\bigr)$ includes both the inference-noise transcript and the encryption-randomness VRF proof. 

The full block-production procedure is shown in Algorithm~\ref{alg:poml-main}. All transactions are fixed by the miner before the first inference, so the binding used by the inference seeds and lottery cannot change during construction.

\begingroup
\setlength{\intextsep}{4pt}
\setlength{\textfloatsep}{4pt}
\begin{algorithm}[tbp]
\caption{PoML Block Production}
\label{alg:poml-main}
\begin{algorithmic}[1]
\Require \textbf{Public:} $B_{\mathsf{prev}}$, difficulty $D$, $\mathsf{CRS}$,
  $M_\theta$, $\theta$, the fee schedule, $G$, $H$, $H_{\mathsf{zk}}$, and
  miner $m$'s registered public keys
\Require \textbf{Private:} $m$'s three secret keys
\State Construct transaction list
  $\mathsf{tx}=(\mathsf{tx}_{\mathsf{std}},\mathcal{T}_{\mathsf{miner}},
  \mathcal{T}_{\mathsf{resp}})$ from the local mempool
\State $s\gets H(B_{\mathsf{prev}})$; $h\gets\mathsf{height}(B_{\mathsf{prev}})+1$
\State $\mathcal{X}\gets()$; $\mathcal{Y}\gets()$;  $\Pi\gets()$;
  $\Pi^{\mathsf{VRF}}\gets()$; $\Omega\gets()$; $i\gets0$
\State $\mathcal{A}\gets$ eligible query objects in mempool whose $\mathsf{qid}$ do not occur in
  $\mathcal{T}_{\mathsf{resp}}$ 
\Repeat
  \State $i\gets i+1$
  \State Select $Q_i$ uniformly at random from $\mathcal{A}$ and remove it from $\mathcal{A}$
  \State $\mathsf{bind}_i\gets
    \begin{cases}G(s,\mathsf{tx})&i=1,\\H(G(s,\mathsf{tx})\parallel\mathcal{Y}_{i-1})&i\geq2\end{cases}$
  \State $r_i\gets H(\mathsf{bind}_i\parallel h_{u,i}\parallel
    \mathsf{qid}_i\parallel\mathsf{vk}^{\mathsf{sig}}_m)$
  \State Retrieve $x_u$ off-chain and verify $Q.h_u = H_\mathsf{zk}(x_u\parallel Q.\mathsf{qid})$
  \For{$t=1,\ldots,T(x_u)$}
    \State $(z_{i,t},\pi^{\mathsf{VRF}}_{i,t})\gets
      \mathsf{VRF.Eval}(\mathsf{sk}^{\mathsf{VRF}}_{\mathsf{inf},m},
        r_i\parallel t)$
  \EndFor
  \State $\mathcal{R}_i\gets(z_{i,1},\ldots,z_{i,T(x_u)})$
  \State $y_i \gets M_\theta(x_i,\mathcal{R}_i)$
  \State $\omega_i\gets C_{\theta,\Pi}(x_i,\mathcal{R}_i)$
  \State $(r_{\mathsf{enc},i},\pi^{\mathsf{VRF}}_{\mathsf{enc},i})\gets
    \mathsf{VRF.Eval}(\mathsf{sk}^{\mathsf{VRF}}_{\mathsf{enc},m},
      r_i\parallel\mathsf{qid}_i)$
  \State $\Pi_i^{\mathsf{VRF}}\gets
    ((z_{i,t},\pi^{\mathsf{VRF}}_{i,t})_{t=1}^{T_i},
      \pi^{\mathsf{VRF}}_{\mathsf{enc},i})$
  \State $\mathsf{ct}_i\gets\mathsf{Enc}(\mathsf{pk}_{u,i},
    y_i\parallel\mathsf{qid}_i;r_{\mathsf{enc},i})$
  \State $\mathsf{stmt}_i\gets(\mathsf{qid}_i,h_{u,i},\theta,\mathcal{R}_i,
    \mathsf{ct}_i,\omega_i,\mathsf{pk}_{u,i},
    \mathsf{vk}^{\mathsf{VRF}}_{\mathsf{enc},m},r_i,
    \pi^{\mathsf{VRF}}_{\mathsf{enc},i})$
  \State $\mathsf{wit}_i\gets(x_{i},y_i,r_{\mathsf{enc},i})$
  \State Append $\mathcal{X}\gets\mathcal{X}\|Q_i$; $\mathcal{Y}\gets\mathcal{Y}\|\mathsf{ct}_i$; $\Pi^{\mathsf{VRF}}\gets\Pi^{\mathsf{VRF}}\|\Pi_i^{\mathsf{VRF}}$
  \State $\mathsf{found}_i\gets
    [H(G(s,\mathsf{tx})\parallel\mathcal{Y})<D_{\omega_i}]$
\Until{$\mathsf{found}_i=\mathsf{true}$}
\For{$j=1,\ldots,i$}
  \State $\pi_j\gets\prover(\mathsf{CRS},\mathsf{stmt}_j,\mathsf{wit}_j)$; $\Pi\gets\Pi\|\pi_j$
\EndFor
\State \textbf{Publish} $B=\bigl(s,\mathsf{tx},\mathsf{VK}_m,\mathcal{X},\mathcal{Y},
  \Pi,\Pi^{\mathsf{VRF}}\bigr)$
\end{algorithmic}
\end{algorithm}
\endgroup

\subsection{Query Responses as Transactions}
\label{sec:response-transactions}

PoML requires miners to compete in producing valid inference outputs. When a miner loses this competition, their work will not be included in the block if there is not a mechanism to `salvage' these work. We therefore allow two routes for responding a query: by including it as part of the consensus work, or by submitting a $\mathsf{RESPONSE\_QUERY}$ transaction. A $\mathsf{RESPONSE\_QUERY}$ transaction simply allows any miner to post their completed work to the blockchain, without using it as a ticket to propose a block.

A response transaction is $T_{\mathsf{resp}}=(Q,r,\mathcal{R},\mathsf{ct},\pi,\Pi^{\mathsf{VRF}},\mathsf{vk}^{\mathsf{sig}}_s,\sigma_s)$, where $s$ is the solver, $\Pi^{\mathsf{VRF}}=((z_t,\pi^{\mathsf{VRF}}_t)_{t=1}^{T},\pi^{\mathsf{VRF}}_{\mathsf{enc}})$, and $\sigma_s\leftarrow\mathsf{Sign}(\mathsf{sk}^{\mathsf{sig}}_s,T_{\mathsf{resp}}\setminus\{\sigma_s\})$. Thus, the complete signed query is available to every validator. The solver may choose the base seed $r$; validators check the resulting inference and proof, but do not require $r$ to be uniform or derived from a previous block. The solver generates a ZK proof before posting the response ciphertext, proof, randomness collection, and VRF transcripts.

At height $h$, a response is valid if:
\begin{enumerate}[nosep]
  \item the query $Q$ included in the transaction is eligible relative to the blockchain state and height $h$, and
  $Q.\mathsf{qid}$ is not completed elsewhere in the current block;
  \item $\mathsf{vk}^{\mathsf{sig}}_s$ is a registered miner and $\sigma_s$ verifies;
  \item for every $t\in[T]$,
    $\mathsf{VRF.Vfy}(\mathsf{vk}^{\mathsf{VRF}}_{\mathsf{inf},s},
      r\parallel t,z_t,\pi^{\mathsf{VRF}}_t)=1$,
    with $\mathcal{R}=(z_1,\ldots,z_T)$;
  \item $\verifier(\mathsf{CRS},\mathsf{stmt}_{\mathsf{resp}},\pi)=1$, where $\mathsf{stmt}_{\mathsf{resp}}$ is the tuple $\bigl(Q.\mathsf{qid},Q.h_u,\theta,\mathcal{R},\mathsf{ct},\omega,\allowbreak Q.\mathsf{pk}_u,\mathsf{vk}^{\mathsf{VRF}}_{\mathsf{enc},s},r,\pi^{\mathsf{VRF}}_{\mathsf{enc}}\bigr)$;
  \item $\phi_h \omega
    \leq F_Q^{\max}$.
\end{enumerate}

Once such a transaction is included in the blockchain, $\mathsf{qid}$ is completed and the fee is debited from the user: the burn component is burned, the solve component is paid to $s$, and the include component is paid to the block producer.

\subsection{Block Validation}
\label{sec:block-validation}

Upon receiving a block $B$, each node checks the following conditions.

\begin{definition}[Valid Block]
\label{def:valid-block}
A block $B=(s,\mathsf{tx},\mathsf{VK}_m,\mathcal{X},\mathcal{Y},
\Pi,\Pi^{\mathsf{VRF}},\Omega)$
is valid with respect to parent $B_{\mathsf{prev}}$ and parent state
$\mathcal{S}_{\mathsf{prev}}$ at candidate height
$h=\mathsf{height}(B_{\mathsf{prev}})+1$ if and only if:
\begin{enumerate}[nosep]
  \item \textbf{Transaction validity:} $s=H(B_{\mathsf{prev}})$,
  $\mathsf{VK}_m$ is active in $\mathcal{S}_{\mathsf{prev}}$, and
  $\mathsf{tx}=(\mathsf{tx}_{\mathsf{std}},\mathcal{T}_{\mathsf{miner}},
  \mathcal{T}_{\mathsf{resp}})$ has the format specified above.  All 
  transactions are valid.
  \item \textbf{Non-empty proof chain:}
  $|\mathcal{X}|=|\mathcal{Y}|=|\Pi|=|\Pi^{\mathsf{VRF}}|=|\Omega|\geq1$;
  \item \textbf{Mining-query eligibility:} every $Q_i\in\mathcal{X}$ is
  eligible relative to $\mathcal{S}_{\mathsf{prev}}$ at height $h$;
  all $\mathsf{qid}$ in $\mathcal{X}$ are pairwise distinct and disjoint from
  those in $\mathcal{T}_{\mathsf{resp}}$;
  \item \textbf{Mining-seed binding:} for every position $i$, validators recompute $\mathsf{bind}_i=G(s,\mathsf{tx})$ for $i=1$ and $\mathsf{bind}_i=H(G(s,\mathsf{tx})\parallel\mathcal{Y}_{i-1})$ for $i\geq2$, then $r_i=H(\mathsf{bind}_i\parallel h_{u,i}\parallel\mathsf{qid}_i\parallel\mathsf{vk}^{\mathsf{sig}}_m)$;
  \item \textbf{Mining inference VRFs:} for every $i$ and $t\in[T_i]$,
   $\mathsf{VRF.Vfy}(\mathsf{vk}^{\mathsf{VRF}}_{\mathsf{inf},m},
     r_i\parallel t,z_{i,t},\pi^{\mathsf{VRF}}_{i,t})=1$,
   with $\mathcal{R}_i=(z_{i,1},\ldots,z_{i,T_i})$;
  \item \textbf{PoML proofs:} for every $i$, validators reconstruct $\mathsf{stmt}_i$ using $Q_i$ and the corresponding entry $\Omega_i$, and check $\verifier(\mathsf{CRS},\mathsf{stmt}_i,\pi_i)=1$;
  \item \textbf{Winning complexity-weighted lottery:} letting $i=|\mathcal{X}|$, the verifier uses the public complexity $\Omega_i$. It then checks $H(G(s,\mathsf{tx})\parallel\mathcal{Y})<D_{\Omega_i}$;
  \item \textbf{Fees:} For each included query, the fee
  $\phi_h\Omega_i$ does not exceed
  $F_{Q_i}^{\max}$ and the user's remaining balance.
\end{enumerate}
\end{definition}

The block producer receives both the solve and include components for every query in
the winning mining chain.  For a valid $\mathsf{RESPONSE\_QUERY}$ transaction,
the separate solver receives the solve component and the block producer receives
the include component.

\subsection{Implementation for Specific Models}
\label{subsec:model-specific-protocol}

We now specify the interpretation of $\mathcal{R}$ and the public complexity function $C_{\theta,\Pi}$ for the model types proposed under PoML.

For diffusion models, $\mathcal{R}$ supplies the inference noise schedule, and $C_{\theta,\Pi} = 1$ for all queries as once the model, dimensions and denoising step count is fixed, all input has the same cost.

For LLMs, $\mathcal{R}$ includes randomness that perturbs the initial prompt embedding and decodes logits into output tokens (Details of perturbation is in App.~\ref{app:llm-perturbation}). For a fixed model, proof system, and setup, an inference--proof cost can usually be approximated from input length $N$ and output length $K$. We therefore define public complexity as $C_{\theta,\Pi}(N,K)$ for both fees and lottery weighting. Appendix~\ref{app:llm-complexity} gives a concrete reference function with GPT-2 using an operation count model similar to gas costs in Ethereum.

At block height $e$, the maximum affordable output length is determined by
$K_{\max}=\max\{K\geq1:\phi_e C_{\theta,\Pi}(N,K)\leq F_Q^{\max}\}$.
If this set is empty, the request is rejected and miners do not select it.
The miner stops when the model emits an end-of-sequence token or after
$K_{\max}$ output tokens, whichever occurs first.

\subsection{Security Reduction: PoML to the Bitcoin Backbone.}
\label{sec:reduction}

We show that PoML inherits the Bitcoin-backbone guarantees of Garay, Kiayias, and Leonardos~\cite{GKL20}---\emph{common prefix}, \emph{chain quality}, and \emph{chain growth}, by showing that PoML similarly inherits the same structure as bitcoin's hash-based proof-of-work. 
The formal assumptions made and detailed argument are in Appendix~\ref{app:reduction}.

\section{Evaluation}
\label{sec:experiments}

We evaluate PoML through simulations using GPT-2 small~\cite{radford2019language} and the DeepProve proving scheme~\cite{lagrange_deepprove}. Synthetic queries arrive regularly (so miners always have work), and each miner runs Algorithm~\ref{alg:poml-main} as a separate process on the same host.
All experiments used one host with a 16-core Intel Core i9 CPU, 128~GB RAM, and an NVIDIA RTX A6000 GPU. Queries use WikiText-2 prompts~\cite{merity2017pointer} of lengths 8--32, and $C_{\theta,\Pi}(N,K)$ follows App.~\ref{subsec:gpt2-reference-operations}. Full implementation details are on GitHub. \footnote{The anonymized version is at \url{https://anonymous.4open.science/r/poml-simulation-730A}.}
We conduct two simulations: one measures block-time stability, and the other measures wasted work.

\subsection{Block-Time Stability}
\label{sec:llm-liveness-experiment}

We choose $300$s as a calibration target and tune PoML's difficulty to our environment, then run it with four honest miners executing the Algorithm~\ref{alg:poml-main} on a pool of 256 WikiText-2 queries. For comparison, we also implement Bitcoin's PoW with the same calibration target. Both chains run until 50 blocks; Fig.~\ref{fig:llm-liveness} reports block intervals and Fig.~\ref{fig:infpair-stats} the relationship between the complexity and inference-proof completion time.

\begin{figure*}[htb]
  \centering
  \begin{subfigure}[t]{0.4\textwidth}
    \centering
    \includegraphics[width=\linewidth]{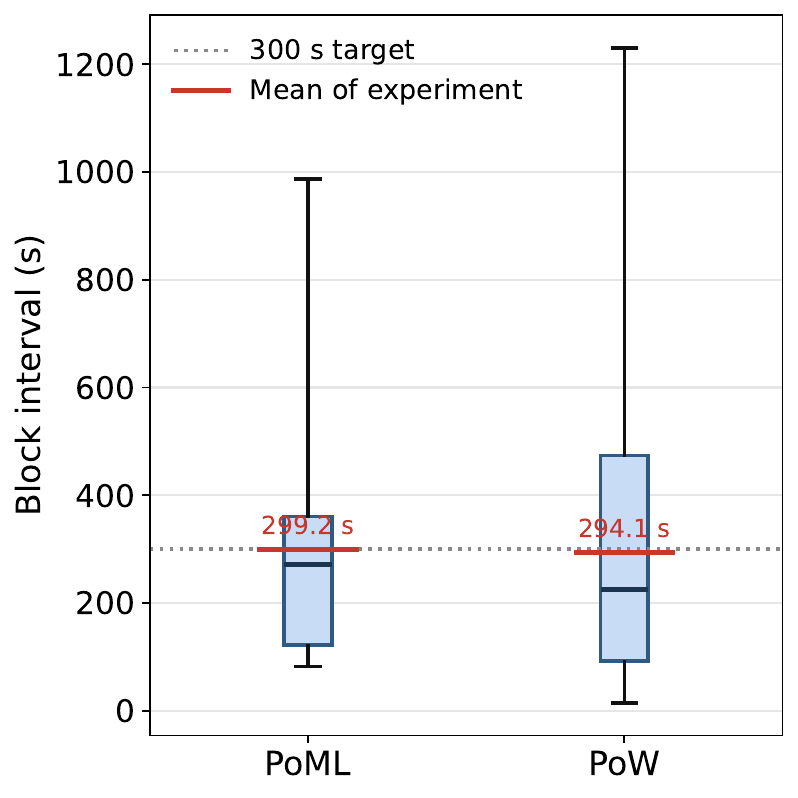}
    \caption{Block intervals.}
    \label{fig:llm-liveness}
  \end{subfigure}\hfill
  \begin{subfigure}[t]{0.4\textwidth}
    \centering
    \includegraphics[width=\linewidth]{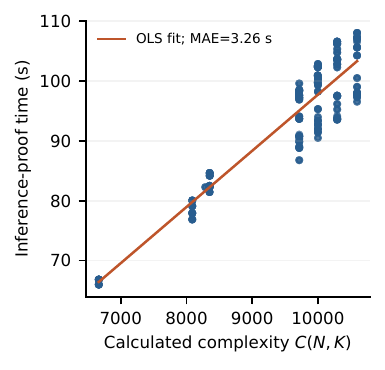}
    \caption{Inference-proof time vs. complexity.}
    \label{fig:infpair-stats}
  \end{subfigure}
  \caption{(a) PoML and PoW block intervals; (b) inference--proof completion time versus calculated query complexity for 546 attempts, with an ordinary least-squares fit. The results show that PoML achieves a stable block production rate and that the complexity function fits the actual runtime with a small error.}
  \Description{Two side-by-side plots show PoML and PoW block intervals, and the relationship between calculated complexity and measured pair completion time with a fitted regression line.}
\end{figure*}

The PoML chain produces all 50 blocks with mean interval 299.19s, close to the 300s target; PoW has a similar mean (294.06s) but larger variance (Fig.~\ref{fig:llm-liveness}). PoML's block intervals are also less variable (standard deviation 220s versus 266s for PoW), likely because inference--proof generation has stable runtimes and PoML's per-attempt lottery success probability is larger. Our complexity function predicts runtime well with mean error 3.26~s (Fig.~\ref{fig:infpair-stats}). Timing errors are unavoidable, but their relative impact should be smaller in real-world systems with higher propagation delays.

\subsection{Wasted Work Analysis}
\label{sec:useful-work-experiment}

Given that a major goal of PoML is to reduce wasteful compute in Bitcoin's PoW, we quantify the amount of wasted work in PoML and compare it with the Bitcoin baseline. Following Pass' formulation~\cite{pass2026pouw}, we write wasted computation as $W=M+D(1-1/\alpha)$, where $M$ is pure mining work, $D$ is duplex work (computation that both helps win the blockchain lottery and produces a useful inference--proof pair) and $\alpha\geq 1$ is the inefficiency ratio ($\alpha$ units of duplex work produce one unit of useful computation), so $D/\alpha$ is useful duplex computation. In Bitcoin, $D=0$ and $W=M$, since all work is mining work. In PoML, $M=0$ and all work belongs in the duplex computation component, hence the waste fraction of such work is $1-1/\alpha$.

To define $\alpha$ in PoML, we consider model inference, proof generation, and proof verification as useful work because a proof of correct execution has inherent value, as evidenced by the growing body of work on verifiable ML~\cite{chen2024zkml,sun2024zkllm,peng2025survey}. For response $i$, let $T_{U,i}$ denote the time spent on these three useful components and let $T_{A,i}$ denote all remaining work, including the VRF computation, the complexity function, and hash evaluations. Across a campaign, we define $T_U=\sum_i T_{U,i}$, $T_A=\sum_i T_{A,i}$, and $T=T_U+T_A$. We then compute $\alpha=T/T_U$, with useful fraction $1/\alpha$ and measured wasted fraction $1-1/\alpha$. This aggregate ratio weights each response by the work it performs rather than averaging ratios from responses with different prompt lengths.

Using the same setup as our first experiment, we measure 50 completed responses with prompt lengths ranging from 8 to 32 tokens. We report the average $\alpha=T/T_U$ ratio. The resulting $\alpha$ is $1.00745$, and the measured wasted fraction is $0.739\%$. Thus, under the same metric Bitcoin has useful fraction zero and wasted fraction one, while PoML's wasted work overhead is below one percent.

Beyond the inefficiency incurred while producing an inference--proof response, PoML can waste work when multiple miners independently process the same query. Miners select queries independently, and the protocol does not restrict which queries they may process. Consequently, two or more miners may produce responses for one query, but only one response can be included in the eventual chain; the remaining completed responses are wasted.

PoML partly reduces this by setting a fixed fee rate, which gives miners no incentive for strategic query selection, so the duplicate work only arises when miners process the same query under random query selection. We measure this duplicate work below.

For one race, let $\mathcal{A}$ contain all completed inference--proof pairs, including duplicates, until a miner wins the lottery, and let $\mathcal{F}\subseteq\mathcal{A}$ contain the first pair for each distinct query. Writing $C_a=C_{\theta,\Pi}(N_a,K_a)$, define wasted work as $W=\bigl(\sum_{a\in\mathcal{A}}C_a-\sum_{a\in\mathcal{F}}C_a\bigr)/\sum_{a\in\mathcal{A}}C_a$. Thus $W$ counts every completed pair beyond the first for a query; with equal complexity it reduces to $(N-U)/N$, where $N$ is the number of pairs and $U$ the number of distinct queries.

We simulate PoML with fixed $Q$ queries and $M$ miners; each miner samples a uniform random permutation and processes it without replacement. We evaluate $M=10$--10,000, $Q=20$--100,000, and target block times 300, 600, and 900~s, repeating each configuration 100 times. Appendix~\ref{app:wasted-work} reports the mean and standard deviation of $W$ across all three target times.

Figures~\ref{fig:uniform-fee-waste-300}--\ref{fig:uniform-fee-waste-900} show that $Q/M$ governs duplicate work: at 300~s, $Q/M=2$ averages 40--42\% wasted work, whereas $Q/M=100$ gives 1.4--1.7\%. At 600 and 900~s, waste increases as miners process more queries before finalisation but remains below 5\% for $Q/M=100$. Thus waste stays low when demand is sufficient and selection is random; and live deployments can control this waste by calibrating block interval to demand and miner supply.

\section{Discussion and conclusion}

\subsubsection{Production-Scale Deployments}
\label{sec:scalability-discussion}
While we only tested PoML on a small example, PoML is feasible on commodity hardware.
Many practical diffusion models have parameter counts in the range of $35$M--$120$M~\cite{ddpm-cifar10-checkpoint,ddpm-celebahq-checkpoint}, well within the operating range of current zkML provers.
As a concrete reference point, zkLLM~\cite{sun2024zkllm} produces an inference--proof pair for a $125$M-parameter model in $73.9$~s on a single A100 GPU.

With a round time of $\Delta = 73.9$~s, $q_{\mathsf{ML}} = 1$ and a per-round success probability $f = 0.10$, the difficulty relation in Section~\ref{subsec:main-thm} gives an expected block time of approximately $12.3$~minutes, roughly comparable to Bitcoin's $10$-minute block target, while tolerating an adversarial fraction up to $\sim 34.6\%$ (full derivation in Appendix~\ref{app:scalability}).

To obtain this block-time estimate, we use a larger honest-success probability $f$ than the value $0.02$ considered in Remark~3 of~\cite{GKL20}. Increasing $f$ shortens the spacing between honest-success rounds, but tightens the condition $3f+3\varepsilon<\gamma$ and increases the probability of multiple honest successes in one round, which can create forks and additional wasted work. Practical deployments must therefore account for the trade-off between $f$ and the adversarial-tolerance bound.

The range of models supported by PoML is also expected to widen substantially as more general and efficient zkML provers become available. Shorter ZKP generation times will allow shorter block times while maintaining sufficient adversarial tolerance. Our work establishes a general principle for PoML consensus that can be adapted to evolving technology.

\subsubsection{Multiple model systems} 
Our construction assumes all miners share a single model. A more interesting regime is a marketplace where miners host different, potentially proprietary, models, as pursued by systems such as Bittensor~\cite{bittensor2023} and Gensyn~\cite{gensyn2022}, though without backbone-style security guarantees. Extending PoML to this setting raises three coupled challenges: \emph{admission}, ensuring registered models satisfy Def.~\ref{def:model} and are not degenerate shortcuts; \emph{valuation}, preventing admissible but useless models from earning rewards, e.g., via peer validation or demand-driven Shapley-style contribution measures~\cite{liu2022gtg}; and \emph{anti-centralisation}, avoiding winner-takes-all query allocation while preserving user choice. Addressing these trustlessly is a substantial direction toward a practical decentralised ML marketplace.

\bibliographystyle{splncs04}
\bibliography{reference}

\appendix
\renewcommand{\theHsection}{app.\Alph{section}}
\renewcommand{\theHsubsection}{app.\Alph{section}.\arabic{subsection}}

\section{PoML System Overview}
\label{app:overview}

The systematic overview diagram is in Fig.~\ref{fig:blockchain-poml}.

\begin{figure*}
    \centering
    \includegraphics[width=0.7\textwidth]{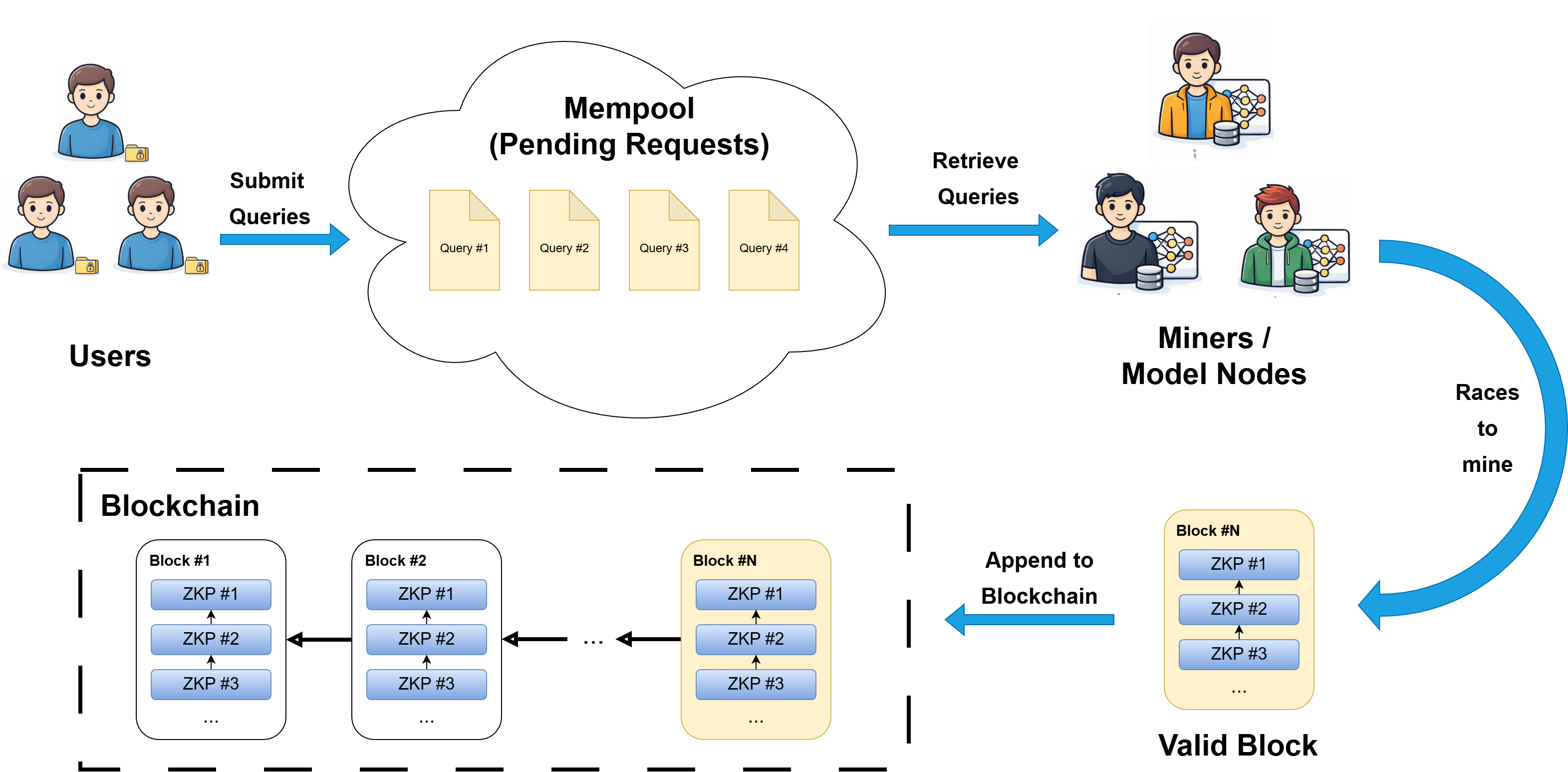}
    \caption{A high-level diagram of the proposed PoML blockchain. Users submit model queries to the mempool as pending requests. Model nodes retrieve these requests and compete to produce a valid block by model inference and generating zero-knowledge proofs. The first node to produce a valid block appends it to the blockchain and receives the fees.}
    \label{fig:blockchain-poml}
\end{figure*}

\section{Security Definitions}

\subsection{Zero-Knowledge Proofs}
\label{app:zk}

\begin{definition}[Non-Interactive Zero-Knowledge proof-of-knowledge Systems] \label{def: NIZK systems}
    A {\em non-interactive zero-knowledge proof-of-knowledge system (NIZKPoK)} for an \NP-language $L$ with the corresponding relation $\relation{L}$
    is a non-interactive protocol $\Pi = (\setup, \prover, \verifier)$, where:
    \begin{itemize}%
        \item $\setup( 1^{\lambda}) \rightarrow \crs$ takes as the input a security parameter $\lambda$. It outputs a common reference string $\crs$.
        \item $\prover(\crs, x, w) \rightarrow \pi$ takes as the input $\crs$, the statement $x$ and the witness $w$, s.t. $(x, w) \in \relation{L}$. It outputs the proof $\pi$.
        \item $\verifier(\crs, x, \pi) \rightarrow b \in \{0,1\}$ takes as the input $\crs$, $x$ and $\pi$. It outputs 1 to accept and 0 to reject.
    \end{itemize}

    $\Pi$ has the following properties:
    \begin{itemize}%
    \item {\bf Completeness.} For all $\lambda \in \NN$, and all $(x,w) \in \relation{L}$, it holds that:
    $$\condprob{\verifier(\crs, x, \prover(\crs, x, w)) = 1}{\crs \sample \setup(1^{|x|}, 1^\lambda)} = 1 - \negligible(\lambda)$$

    \item {\bf Soundness.} For all $\ppt$ provers $\prover^{\star}$, s.t. for all $\lambda \in \NN$, and all $x\notin L$, it holds that:
    $$\condprob{\verifier(\crs,x,\pi) = 1}{\crs \sample \setup(1^{|x|}, 1^{\lambda}); \pi \sample \prover^{\star}(\crs)} \leq \negligible(\lambda).$$

    \item {\bf Zero knowledge.} There exists a $\ppt$ simulator $\simu$ such that
    for every $(x,w)\in \relation{L}$, the distribution ensembles
    $\{(\crs,\pi): \crs \sample \setup(1^{|x|}, 1^{\lambda});$ $\pi \sample \prover(\crs, x, w)\}_{\lambda \in \NN}$
    and $\{\simu(1^{\lambda},x)\}_{\lambda\in\NN}$
    are computationally indistinguishable.

    \item {\bf Proof of Knowledge.} For all $\ppt$ provers $\prover$, there exists an extractor $\mathcal{E}$ such that
    $$\condprob{\mathcal{E}(x,\prover) = w}{(x,w) \in \relation{L}} = 1 - \negligible(\lambda)$$
    \end{itemize}

    \end{definition}

\subsection{Public-Key Encryption}
\label{app:pke}

\begin{definition}[Public-Key Encryption Scheme]
\label{def:pke-app}
    A \emph{public-key encryption (PKE) scheme} over a message space $\mathcal{X}$ is a tuple $\Sigma = (\mathsf{KeyGen}_{\mathsf{pke}}, \mathsf{Enc}, \mathsf{Dec})$ of $\ppt$ algorithms where:
    \begin{itemize}
        \item $\mathsf{KeyGen}_{\mathsf{pke}}(1^{\lambda}) \rightarrow (\mathsf{pk}, \mathsf{sk})$ takes the security parameter $\lambda$ and outputs a key pair consisting of a public key $\mathsf{pk}$ and secret key $\mathsf{sk}$.
        \item $\mathsf{Enc}(\mathsf{pk}, m; r) \rightarrow \mathsf{ct}$ encrypts a message $m \in \mathcal{X}$ under public key $\mathsf{pk}$ using randomness $r \in \{0,1\}^\kappa$, and outputs a ciphertext $\mathsf{ct}$.
        \item $\mathsf{Dec}(\mathsf{sk}, \mathsf{ct}) \rightarrow m$ decrypts ciphertext $\mathsf{ct}$ under secret key $\mathsf{sk}$ and returns the recovered message.
    \end{itemize}

    $\Sigma$ has the following properties:
    \begin{itemize}
    \item {\bf Correctness.} For all $\lambda \in \NN$, all $m \in \mathcal{X}$, and all $r \in \{0,1\}^\kappa$, it holds that:
    \[
        \Pr\!\left[
            \mathsf{Dec}(\mathsf{sk},\, \mathsf{Enc}(\mathsf{pk}, m; r)) = m :
            (\mathsf{pk},\mathsf{sk}) \leftarrow \mathsf{KeyGen}_{\mathsf{pke}}(1^{\lambda})
        \right] = 1.
    \]

    \item {\bf IND-CPA Security.} For all $\ppt$ adversaries $\mathcal{A} = (\mathcal{A}_1, \mathcal{A}_2)$, it holds that:
    \[
        \left|\Pr\!\left[
            b' = b :
            \begin{array}{l}
                (\mathsf{pk}, \mathsf{sk}) \leftarrow \mathsf{KeyGen}_{\mathsf{pke}}(1^{\lambda}) \\
                (m_0, m_1, \mathsf{st}) \leftarrow \mathcal{A}_1(\mathsf{pk}) \\
                b \overset{\$}{\leftarrow} \{0,1\} \\
                r \overset{\$}{\leftarrow} \{0,1\}^\kappa \\
                \mathsf{ct} \leftarrow \mathsf{Enc}(\mathsf{pk}, m_b; r) \\
                b' \leftarrow \mathcal{A}_2(\mathsf{pk}, \mathsf{ct}, \mathsf{st})
            \end{array}
        \right] - \frac{1}{2}\right| \leq \negligible(\lambda).
    \]
    \end{itemize}
\end{definition}

\subsection{Verifiable Random Functions and Unbiasability}
\label{app:vrf}

We recall the formal definition of a VRF~\cite{micali1999vrf} and the notion of unbiasability introduced by Giunta and Stewart~\cite{unbiasedVRF}. Brandt~\cite{unbiasedVRFForm} shows that unbiasable VRFs can be built from generic cryptographic assumptions, i.e.\ \emph{without} relying on a random oracle. For inputs $x \in X$ and outputs $y \in Y$ (over some codomain $Y$, e.g.\ a finite field $\mathbb{F}_q$), we use the following definitions.

\begin{definition}[Verifiable Random Function]
\label{def:vrf-app}
A \emph{Verifiable Random Function} is a triplet of $\ppt$ algorithms $(\mathsf{VRF.Gen},$ $\mathsf{VRF.Eval},$ $\mathsf{VRF.Vfy})$ satisfying:
\begin{itemize}
    \item \textbf{Correctness.} For any $(\vk, \sk)$ in the image of $\mathsf{VRF.Gen}(1^\lambda)$ and any $x \in X$,
    \[
        (y, \pi) \leftarrow \mathsf{VRF.Eval}(\sk, x) \;\Rightarrow\; \mathsf{VRF.Vfy}(\vk, x, y, \pi) \to 1.
    \]
    \item \textbf{Unique Provability.} For any $\vk$ (not necessarily in the range of $\mathsf{VRF.Gen}$), any input $x \in X$, any pair of outputs $y_0, y_1 \in Y$, and any pair of proofs $\pi_0, \pi_1$,
    \[
        \begin{aligned}
        &\mathsf{VRF.Vfy}(\vk, x, y_0, \pi_0) \to 1 \,\wedge\, \mathsf{VRF.Vfy}(\vk, x, y_1, \pi_1) \to 1 \\
        &\qquad\Rightarrow\; y_0 = y_1.
        \end{aligned}
    \]
    \item \textbf{Pseudorandomness.} For any $\ppt$ adversary $\adv$ playing the game in Fig.~\ref{fig:vrf-prand}, there exists a negligible function $\varepsilon$ such that
    \[
        \mathsf{Adv}(\adv) \;=\; \Bigl| \Pr\!\left[\mathsf{Exp}^{\mathsf{rnd}}(\adv) \to 1\right] - \tfrac{1}{2} \Bigr| \;\leq\; \varepsilon(\lambda).
    \]
\end{itemize}
\end{definition}

\begin{figure}[t]
\centering
\fbox{\parbox{0.92\columnwidth}{%
\footnotesize
\textbf{$\mathsf{Exp}^{\mathsf{rnd}}(\adv)$:}
\begin{enumerate}[leftmargin=*,itemsep=0em,nosep]
    \item Sample $b \sample \{0,1\}$.
    \item $(\vk, \sk) \sample \mathsf{VRF.Gen}(1^\lambda)$; set $x^\star := \perp$.
    \item $x^\star \leftarrow \adv^{\mathcal{O}_{\mathsf{eval}}}(\vk)$.
    \item If $x^\star$ was previously queried, return $0$.
    \item $y_0 \sample Y$.
    \item $(y_1, \pi) \leftarrow \mathsf{VRF.Eval}(\sk, x^\star)$.
    \item $b' \leftarrow \adv^{\mathcal{O}_{\mathsf{eval}}}(\vk, y_b)$.
    \item Return $\mathbf{1}[b = b']$.
\end{enumerate}
\vspace{0.3em}
\textbf{$\mathcal{O}_{\mathsf{eval}}(x)$:} \; if $x \neq x^\star$, return $\mathsf{VRF.Eval}(\sk, x)$; else return $\perp$.
}}
\caption{The pseudorandomness security game for VRFs.}
\label{fig:vrf-prand}
\end{figure}

\paragraph{Unbiasability.}
Unbiasability asks that even an adversary that controls some of the keys and selectively withholds evaluations cannot bias the joint output distribution on uniformly chosen challenge inputs.

\begin{definition}[Unbiasable VRF~\cite{unbiasedVRF}]
\label{def:unbiasable-app}
A VRF is \emph{unbiasable} if for any polynomially bounded integers $n, m$ and any $\ppt$ adversary $\adv$ there exists a negligible function $\varepsilon$ such that
\[
    \mathsf{Adv}(\adv) := \Pr\!\left[\mathsf{Exp}^{\mathsf{bias}}_{0,n,m}(\adv) = 1\right] - \Pr\!\left[\mathsf{Exp}^{\mathsf{bias}}_{1,n,m}(\adv) = 1\right] \leq \varepsilon(\lambda),
\]
where the experiment $\mathsf{Exp}^{\mathsf{bias}}_{b,n,m}$ is defined in Fig.~\ref{fig:vrf-bias}.
\end{definition}

\begin{figure}[t]
\centering
\fbox{\parbox{0.92\columnwidth}{%
\footnotesize
\textbf{$\mathsf{Exp}^{\mathsf{bias}}_{b,n,m}(\adv)$:}
\begin{enumerate}[leftmargin=*,itemsep=0em,nosep]
    \item $(\vk_1, \ldots, \vk_n, p) \leftarrow \adv$, where
    \begin{itemize}[leftmargin=*,itemsep=0em,nosep]
        \item $\vk_1, \ldots, \vk_n$ are all distinct;
        \item $p$ is a $\ppt$-computable monotone predicate.
    \end{itemize}
    \item $(x_1, \ldots, x_m) \sample X$.
    \item $\bigl((y_{1,1}, \pi_{1,1}), \ldots, (y_{n,m}, \pi_{n,m})\bigr) \leftarrow \adv(x_1, \ldots, x_m)$.
    \item For all $i \in [n]$, $j \in [m]$: if $\mathsf{VRF.Vfy}(\vk_i, x_j, y_{i,j}, \pi_{i,j}) \to 0$, set $y_{i,j} \leftarrow \perp$.
    \item Sample $(z_{1,1}, \ldots, z_{n,m}) \sample Y^{n \cdot m}$.
    \item If $b = 0$, return $p(y_{1,1}, \ldots, y_{n,m})$.
    \item If $b = 1$, return $p(z_{1,1}, \ldots, z_{n,m})$.
\end{enumerate}
}}
\caption{The unbiasability experiment, parametrized by $b \in \{0,1\}$.}
\label{fig:vrf-bias}
\end{figure}

Giunta and Stewart~\cite{unbiasedVRF} prove that several concrete constructions meet Definition~\ref{def:unbiasable-app}, and Brandt~\cite{unbiasedVRFForm} subsequently showed a construction from generic assumptions without random oracles. For our purposes it is therefore safe to assume the existence of an unbiasable VRF in the standard model.

\subsection{Digital Signature Schemes}
\label{app:sig}

\begin{definition}[Digital Signature Scheme]
\label{def:sig-app}
    A \emph{digital signature scheme} over a message space $\mathcal{M}$ is a tuple $\Sigma_{\mathsf{sig}} = (\mathsf{KeyGen}_{\mathsf{sig}}, \mathsf{Sign}, \mathsf{Verify}_{\mathsf{sig}})$ of $\ppt$ algorithms where:
    \begin{itemize}
        \item $\mathsf{KeyGen}_{\mathsf{sig}}(1^{\lambda}) \rightarrow (\mathsf{vk}, \mathsf{sk})$ takes the security parameter $\lambda$ and outputs a verification key $\mathsf{vk}$ and a signing key $\mathsf{sk}$.
        \item $\mathsf{Sign}(\mathsf{sk}, m) \rightarrow \sigma$ takes the signing key $\mathsf{sk}$ and a message $m \in \mathcal{M}$, and outputs a signature $\sigma$.
        \item $\mathsf{Verify}_{\mathsf{sig}}(\mathsf{vk}, m, \sigma) \rightarrow b \in \{0,1\}$ takes the verification key $\mathsf{vk}$, a message $m \in \mathcal{M}$, and a signature $\sigma$, and outputs $1$ to accept or $0$ to reject.
    \end{itemize}

    $\Sigma_{\mathsf{sig}}$ has the following properties:
    \begin{itemize}
    \item {\bf Correctness.} For all $\lambda \in \NN$, all $m \in \mathcal{M}$, it holds that:
    \[
        \Pr\!\left[
            \mathsf{Verify}_{\mathsf{sig}}(\mathsf{vk},\, m,\, \mathsf{Sign}(\mathsf{sk}, m)) = 1 :
            (\mathsf{vk}, \mathsf{sk}) \leftarrow \mathsf{KeyGen}_{\mathsf{sig}}(1^{\lambda})
        \right] = 1.
    \]

    \item {\bf EUF-CMA Security.} For all $\ppt$ adversaries $\mathcal{A}$, it holds that:
    \[
        \Pr\!\left[
            \begin{array}{l}
                \mathsf{Verify}_{\mathsf{sig}}(\mathsf{vk}, m^\star, \sigma^\star) = 1 \\
                \wedge\; m^\star \notin Q
            \end{array}
            :
            \begin{array}{l}
                (\mathsf{vk}, \mathsf{sk}) \leftarrow \mathsf{KeyGen}_{\mathsf{sig}}(1^{\lambda}) \\
                (m^\star, \sigma^\star) \leftarrow \mathcal{A}^{\mathcal{O}_{\mathsf{sign}}}(\mathsf{vk})
            \end{array}
        \right] \leq \negligible(\lambda),
    \]
    where $\mathcal{O}_{\mathsf{sign}}(m)$ returns $\mathsf{Sign}(\mathsf{sk}, m)$ and records each queried message in the set $Q$.
    \end{itemize}
\end{definition}

\section{LLM Compatibility: Formal Statements and Experiments}
\label{app:llm-perturbation}

We define the perturbed LLM construction and argue it satisfies computational independence. The experiments use GPT-2~\cite{radford2019language}, with 12 Transformer blocks and hidden width $d=768$, in deterministic evaluation mode with float32 arithmetic.

\subsection{Formal Description of the LLM Construction}
\label{app:llm-formal-description}

\begin{definition}[LLM with perturbed prompt embeddings]
\label{def:llm-perturbed}
A perturbed autoregressive LLM is specified by
\[
 M_\theta=(f_\theta,E_\theta,P,\mathcal V,d,L,s_{\max},\tau,k,p,m_{\max},\alpha).
\]
Here $\mathcal V$ is the vocabulary, $E_\theta\in\mathbb R^{|\mathcal V|\times d}$ is the token-embedding table, $P_i\in\mathbb R^d$ is the positional embedding at position $i$, and
\[
 f_\theta:\bigcup_{n=1}^{s_{\max}}\mathbb R^{n\times d}
       \longrightarrow\mathbb R^{|\mathcal V|}
\]
is the fixed causal Transformer, including its final normalisation and output head. It applies $L$ Transformer blocks to an embedded context and returns the next-token logits from its last position. Here $d$ is the embedding and hidden-vector width. The remaining parameters are the context limit $s_{\max}$, decoding settings (temperature $\tau>0$, top-$k$ and top-$p$), maximum output length $m_{\max}$, and noise parameter $\alpha>0$.

For a tokenised prompt $x=(x_1,\ldots,x_s)$ with $s+m_{\max}\leq s_{\max}$, form the perturbed prompt embeddings as follows:
\begin{equation}
\begin{aligned}
    \widetilde e_{r,i}&=E_\theta(x_i)+P_i+\eta_{r,i},\\
    \eta_{r,i}&=\alpha s_E z_{r,i},\qquad
    z_{r,i}\overset{\mathrm{iid}}{\sim}\mathcal N(0,I_d),\quad i\in[s].
\end{aligned}
\label{eq:llm-noise-law}
\end{equation}
The table standard deviation $s_E=\operatorname{std}\{(E_\theta)_{v,j}:v\in\mathcal V,\ j\in[d]\}$ is computed once from the public model, and $\sigma_{\mathrm{abs}}=\alpha s_E$.  Noise is added once to the $s$ prompt embeddings, giving $sd$ perturbed coordinates.

Let $C_0(r)=(\widetilde e_{r,1},\ldots,\widetilde e_{r,s})$ and let $y_{1:t}(r)$ denote the first $t$ generated tokens. The decoding recurrence is
\begin{equation}
\begin{aligned}
 a_t(r)&=f_\theta(C_t(r)),\\
 y_{t+1}(r)&=\mathsf{Sample}_{\tau,k,p}(a_t(r);u_t(r)),\\
 C_{t+1}(r)&=C_t(r)\mathbin\Vert
       \bigl(E_\theta(y_{t+1}(r))+P_{s+t+1}\bigr),
\end{aligned}
\label{eq:llm-decoding-recurrence}
\end{equation}
for $t=0,\ldots,m_{\max}-1$, stopping at EOS. Here $\Vert$ appends a row and $\mathsf{Sample}_{\tau,k,p}$ applies temperature scaling, top-$k$/top-$p$ filtering, and inverse-CDF sampling with uniform value $u_t(r)$.

For the computational independence argument, we consider the sequence of pre-decoding logit vectors as the output of $M_\theta$:
\[
 M_\theta(x,\mathcal R)=(a_0(r),\ldots,a_{m-1}(r)),\qquad m\leq m_{\max},
\]
The decoding of the result itself is a simple lookup operation, hence its cost is negligible compared to the rest of the operations.

The VRF-derived randomness collection $\mathcal{R}(r)$ reserves $s+m_{\max}$ indexed values: the first $s$ supply the vectors $z_{r,i}$ through a public expansion, and the remaining values supply $u_t(r)$. Values reserved for tokens after termination are unused. Given the prompt and $\mathcal{R}(r)$, inference is deterministic.
\end{definition}

\subsection{Preservation of Benchmark Utility}
\label{app:llm-utility}

To confirm that the noise perturbation does not affect generation quality, we compare the clean model (one without perturbation) with the perturbed model on WikiText-2~\cite{merity2017pointer}, HellaSwag~\cite{zellers2019hellaswag}, PIQA~\cite{bisk2020piqa}, and ARC-Easy~\cite{clark2018arc}, using 100 sampled examples per dataset and three perturbation seeds. WikiText-2 is evaluated by perplexity, computed by teacher-forced scoring. Other datasets are evaluated by multiple-choice accuracy.

Both experiments use $\alpha\in\{0.05,0.10,0.20,0.40\}$ and $\sigma_{\mathrm{abs}}=\alpha s_E$, with $s_E=0.14370$ measured from the frozen embedding table. 

\begin{table}[t]
\centering
\small
\caption{GPT-2 utility under Gaussian prompt-embedding noise $\sigma_{\mathrm{abs}}=\alpha s_E$. Perplexity (PPL) is lower-is-better; accuracy is in percent and higher-is-better. Each task uses 100 examples; perturbed results average three seeds.}
\label{tab:llm-embedding-utility}
\begin{tabular}{lrrrrr}
\toprule
 & Clean & \multicolumn{4}{c}{$\alpha$} \\
\cmidrule(lr){3-6}
Task / metric & & $0.05$ & $0.10$ & $0.20$ & $0.40$ \\
\midrule
WikiText-2 / PPL & 44.80 & 44.79 & 44.77 & 45.39 & 44.14 \\
HellaSwag / accuracy & 25.00 & 25.00 & 24.67 & 24.67 & 24.33 \\
PIQA / accuracy & 62.00 & 62.00 & 63.00 & 61.33 & 62.33 \\
ARC-Easy / accuracy & 36.00 & 36.67 & 36.67 & 36.33 & 36.33 \\
\bottomrule
\end{tabular}
\end{table}

Table~\ref{tab:llm-embedding-utility} shows the results. At $\alpha=0.05$, WikiText-2 perplexity changes from $44.80$ to $44.79$ ($-0.011\%$), HellaSwag and PIQA accuracy are unchanged, and ARC-Easy increases by $0.67$ percentage points. Across the tested scales, the largest perplexity increase is $1.31\%$ and the largest accuracy decrease is $0.67$ percentage points. These results support preservation of benchmark utility on the sampled examples under modest embedding noise.

\subsection{Computational Independence of LLMs}
\label{app:llm-trace-separation}

To establish the computational independence of the input-perturbed LLM, we adopt the following assumption:

\begin{assumption}[Computational Independence of $f_\theta$]
\label{ass:llm-no-reuse}
Fix the Transformer $f_\theta$ of Definition~\ref{def:llm-perturbed}. For a fixed prompt $x$, let $h_{\ell,t}(r)$ be the output matrix of layer $\ell$ during $f_\theta(C_t(r))$. For independent protocol challenges $r,r'$ and each step $t$ executed by both runs, assume
\begin{equation}
 \Pr\!\left[\min_{\ell\in[L]}
       \|h_{\ell,t}(r)-h_{\ell,t}(r')\|_\infty
       <\delta_q/2\right]\leq\mathsf{negl}(d),
 \label{eq:llm-block-separation}
\end{equation}
where $d$ is the model's hidden-vector width.
\end{assumption}

To stress-test this assumption, we use 100 prompts each from WikiText-2~\cite{merity2017pointer}, LAMBADA~\cite{paperno2016lambada}, HellaSwag~\cite{zellers2019hellaswag}, PIQA~\cite{bisk2020piqa}, ARC-Easy~\cite{clark2018arc}, and Resisting Correction~\cite{mckenzie2023inverse}. The Resisting Correction dataset contains prompts which require the model to repeat certain phrases, trying to deliberately skew the model's output logits. We use the same relative Gaussian noise as the utility experiment, with four independent challenge pairs per prompt, running GPT-2 in float32 on an NVIDIA RTX A6000.

For each pair of inferences with the same prompt and different randomness, we decode tokens as usual and measure the $L^\infty$ distance between the two output activations at intermediate layers of $f_\theta$.

\begin{figure}[t]
\centering
\begin{tikzpicture}
\begin{semilogyaxis}[width=\linewidth,height=5.2cm,
xlabel={Layer output (A: attention; F: feed-forward)},
ylabel={Minimum $L^\infty$ distance},
tick label style={font=\scriptsize},label style={font=\small},
xtick={0,1,2,3,4,5,6,7,8,9,10,11,12,13,14,15,16,17,18,19,20,21,22,23,24,25,26},
xticklabels={Emb.,1A,1F,2A,2F,3A,3F,4A,4F,5A,5F,6A,6F,7A,7F,8A,8F,9A,9F,10A,10F,11A,11F,12A,12F,Norm,Logits},
x tick label style={rotate=60,anchor=east},
grid=major,grid style={gray!15},
legend style={font=\scriptsize,at={(0.5,1.02)},anchor=south,
legend columns=3,draw=none}]
\addplot[blue,mark=*,mark size=1pt] coordinates {(0,0.02975048684) (1,0.1315028667) (2,0.1722182035) (3,0.06826353073) (4,0.08200213313) (5,0.008858084679) (6,0.06993815303) (7,0.002136945724) (8,0.06260681152) (9,0.002312481403) (10,0.01658630371) (11,0.002029299736) (12,0.01611328125) (13,0.002088546753) (14,0.009460449219) (15,0.002996891737) (16,0.006938934326) (17,0.00261053443) (18,0.006219863892) (19,0.004909187555) (20,0.005334854126) (21,0.005473516881) (22,0.006958961487) (23,0.0478515625) (24,0.09757894278) (25,0.1108818054) (26,0.01602172852)};
\addlegendentry{$\alpha=0.05$}
\addplot[orange,mark=square*,mark size=1pt] coordinates {(0,0.05904448032) (1,0.2998138666) (2,0.278876543) (3,0.08347073197) (4,0.1852707863) (5,0.01018296182) (6,0.1690934002) (7,0.005183987319) (8,0.1734703183) (9,0.003881677985) (10,0.1094207764) (11,0.008554577827) (12,0.08500671387) (13,0.009910583496) (14,0.05267715454) (15,0.005917727947) (16,0.03393936157) (17,0.00998544693) (18,0.01668930054) (19,0.01132512093) (20,0.01330852509) (21,0.01852703094) (22,0.02373886108) (23,0.5471191406) (24,0.2649250031) (25,0.1517615318) (26,0.03082275391)};
\addlegendentry{$\alpha=0.1$}
\addplot[green!50!black,mark=triangle*,mark size=1pt] coordinates {(0,0.1246514954) (1,0.6116398573) (2,0.6193690896) (3,0.1694792509) (4,0.4277038574) (5,0.02490875125) (6,0.1215820312) (7,0.00587439537) (8,0.3810467124) (9,0.007812231779) (10,0.1911468506) (11,0.008918553591) (12,0.1309509277) (13,0.01127564907) (14,0.07312774658) (15,0.01688092947) (16,0.05374145508) (17,0.01252047718) (18,0.03097343445) (19,0.02128940821) (20,0.0368013382) (21,0.02147084475) (22,0.04918050766) (23,0.9663085938) (24,0.3410863876) (25,0.3962593079) (26,0.06625366211)};
\addlegendentry{$\alpha=0.2$}
\addplot[red,mark=diamond*,mark size=1pt] coordinates {(0,0.2423552796) (1,1.374424338) (2,1.208605111) (3,0.2650289461) (4,0.9603271484) (5,0.07097357512) (6,0.1892089844) (7,0.01133602858) (8,0.4913482666) (9,0.01510684192) (10,0.1653289795) (11,0.01869440079) (12,0.1112213135) (13,0.03177942336) (14,0.109172821) (15,0.02454146743) (16,0.02603912354) (17,0.02328563482) (18,0.03371432424) (19,0.04424023628) (20,0.04525583982) (21,0.06023013592) (22,0.09832429886) (23,0.8720703125) (24,1.185741425) (25,0.7682647705) (26,0.1420593262)};
\addlegendentry{$\alpha=0.4$}
\addplot[black,dashed,no marks] coordinates {(0,0.001) (26,0.001)};
\addlegendentry{Quantization width $\Delta=10^{-3}$}
\end{semilogyaxis}
\end{tikzpicture}
\caption{Minimum $L^\infty$ distance at each GPT-2 layer output. Each point is the minimum of 14,400 comparisons (600 prompts, four pairs, six prefixes). All shown minima exceed the comparison quantization-bin width $\Delta=10^{-3}$.}
\label{fig:llm-layer-linf}
\end{figure}
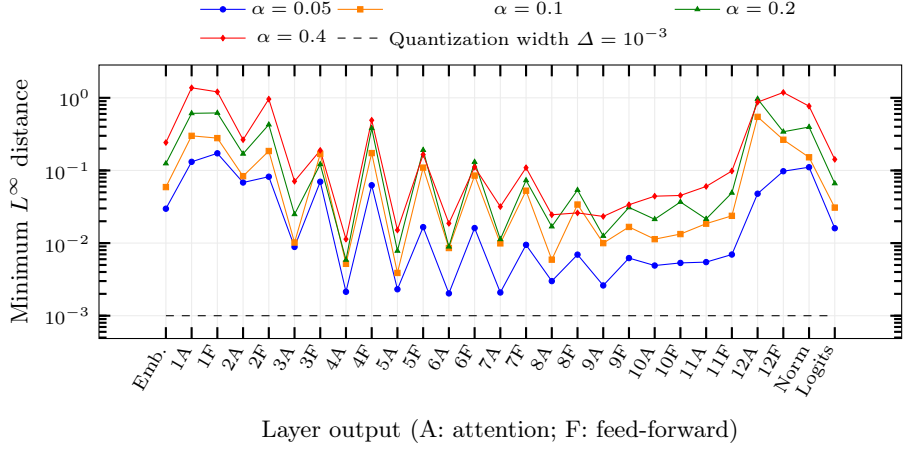

Figure~\ref{fig:llm-layer-linf} shows the minimum recorded $L^\infty$ distance at each layer. At $\alpha=0.05$, the smallest distance is $2.0293\times10^{-3}$ (block 6 attention), exceeding both $\Delta$ and the proximity tolerance $\delta_q/2=5\times10^{-4}$ for $\delta_q=10^{-3}$. These observations provide limited evidence against direct reuse of cached activation tensors in the tested settings; they do not establish computational independence. A fuller theoretical and empirical justification of this assumption is left to future work.

\subsubsection{Generated-prefix agreement}
\label{subsec:llm-prefix-collisions}

We also measure whether independently sampled runs happen to agree on their generated tokens. We use 500 deduplicated 32-token WikiText-2 test prompts and temperatures $\tau\in\{0.7,1.0,1.3,1.5,2.0\}$. We perform the experiment 4 times per prompt, each time generating a pair of output sequences by sampling the output independently. Each pair uses VRF-derived sampling uniforms and full categorical sampling (no top-$k$, top-$p$, or repetition-penalty filtering), with a maximum of 32 generated tokens and early stopping at EOS. For a pair of decoded traces $Y^{(a)}$ and $Y^{(b)}$, we estimate
\[
  C_\ell = \Pr\!\left[Y^{(a)}_{1:\ell}=Y^{(b)}_{1:\ell}\right]
\]
from pairs for which both traces agree until position $\ell$. Thus $C_\ell$ is the probability that two traces share their first $\ell$ generated tokens, and $\sum_{\ell\geq 1} C_\ell$ is the expected shared-prefix length. This experiment complements the activation-distance experiment above and shows that even with the same prompt, it is highly unlikely for two independent inferences to produce the exact same result. 

\begin{figure*}[t]
  \centering
  \includegraphics[width=0.96\textwidth]{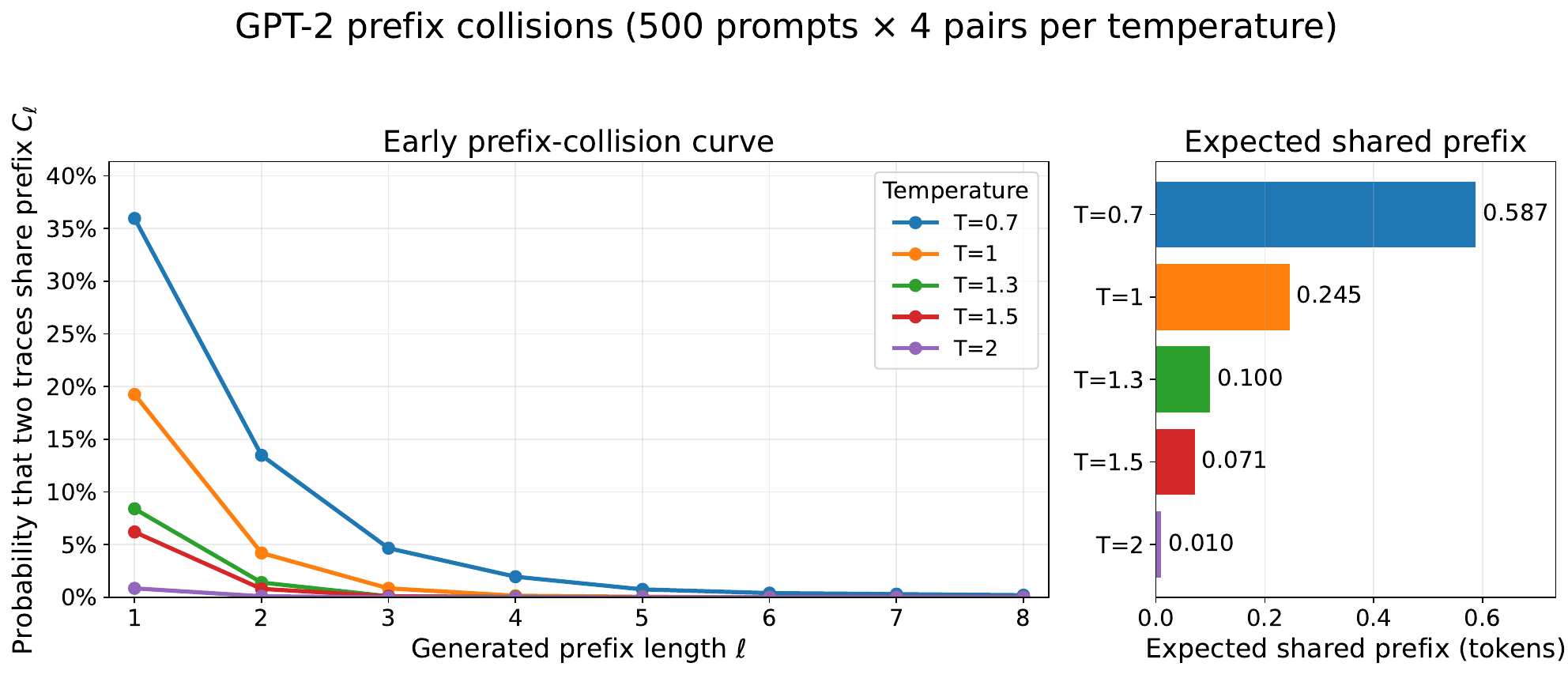}
  \caption{Exact generated-prefix collision probability for independent GPT-2 traces. The experiment uses 500 WikiText-2 prompts, four pairs per prompt, and a 32-token generation cap; pairs are censored at each prefix length when either trace stops at EOS. The right panel reports $\sum_{\ell\geq1}C_\ell$, the expected shared-prefix length.}
  \Description{Two panels show the probability that independently sampled GPT-2 traces share each generated prefix length at several temperatures, together with the expected shared-prefix length for each temperature.}
  \label{fig:llm-prefix-collisions}
\end{figure*}

\begin{theorem}[Computational independence of the logit output]
\label{thm:llm-computational-independence}
Under Assumption~\ref{ass:llm-no-reuse}, the construction of Definition~\ref{def:llm-perturbed} satisfies Definition~\ref{def:delta-reuse} with parameter $\varepsilon_i$ accounting only a small embedding and unembedding process.
\end{theorem}

\begin{proof}
An adversary satisfying the depth bound in Definition~\ref{def:delta-reuse} must reproduce the exact logit sequence. Assumption~\ref{ass:llm-no-reuse} bounds its success by a negligible probability, as any adversary cannot reuse any intermediate vectors of earlier computation and will have to recompute the exact vectors based on the assigned randomness. 
\end{proof}

\section{DDPM Compatibility: Formal Statements and Experiments}
\label{app:cia}

\subsection{Formal Description of DDPM}
\label{subsec:ddpm-formal}

We instantiate the stochastic denoising procedure of Section~\ref{subsec:diffusion-models} as follows.

\begin{definition}[DDPM]
\label{def:stochastic-denoising}
A \emph{DDPM} is a tuple
\[
        M_\theta = \bigl(f_\theta, T, (\sigma_t)_{t=2}^{T}, d\bigr),
\]
where $f_\theta : \mathbb{R}^d \times \mathbb{R}^{d_\mathsf{u}} \times [T] \to \mathbb{R}^d$ is the learned denoising model, $T \geq 2$ is the number of denoising steps, each $\sigma_t>0$ is a fixed noise standard deviation, and $d$ is the latent dimension.
\end{definition}

The inputs are the conditioning embedding $x_\mathsf{u}=e(p)$ of the user's prompt and the noise schedule $\mathcal{R}=(z_1,\ldots,z_T)$. We set $x_T=z_1$, use $z_t$ at reverse step $t$ for $t=T,\ldots,2$, and compute the output $y=x_0$ using the DDPM recurrence from Section~\ref{subsec:diffusion-models}, followed by its noiseless final step. Given $(x_\mathsf{u},\mathcal{R})$, the evaluation is deterministic.

\subsection{Computational Independence}
\label{subsec:ddpm-properties}

The argument has two parts: in each denoising step, fresh noise is added which makes any input distinct with high probability, and an additional assumption bounds reuse within the denoising model. We use $\delta_q>0$ to specify the numerical comparison tolerance, calling two tensors close when their $L^\infty$ distance is below $\delta_q/2$. This is such that in the ZKP, if its quantization precision is at least $\delta_q$, the two tensors are considered different.

\subsubsection{Distinct Inputs at Every Step}
\label{app:input-divergence}

We now formalise the first stage of the $\varepsilon_i$-computational-independence argument: for two runs of $M_\theta$ with independent seeds, the inputs to $f_\theta$ at every denoising step differ by far more than $\delta_q/2$ in $L^\infty$ for some fixed $\delta_q$. The proof uses standard Gaussian anti-concentration.

\begin{lemma}[Gaussian anti-concentration]
\label{lem:anti-conc-app}
Let $W \sim \mathcal{N}(0, \sigma^2)$ with $\sigma > 0$. For any $a \in \mathbb{R}$ and any $\delta_q > 0$:
\[
    \Pr\!\bigl[|a + W| \leq \delta_q/2\bigr] \;\leq\; \frac{\delta_q}{\sigma\sqrt{2\pi}}.
\]
\end{lemma}

\begin{proof}
The density of $a + W$ is bounded above by $\frac{1}{\sigma\sqrt{2\pi}}$. Integrating over an interval of width $\delta_q$ gives the result.
\end{proof}

\begin{theorem}[Step-level input differences in stochastic denoising]
\label{thm:cia-theorem}
Let $M_\theta$ be a DDPM (Definition~\ref{def:stochastic-denoising}) with $T$ denoising steps, latent dimension~$d$, and noise standard deviations $\sigma_t > 0$ for $t = T, \ldots, 2$. Let $\sigma_{\min} = \min\{1,\sigma_2,\ldots,\sigma_T\}$ and suppose $\delta_q < 2\sigma_{\min}\sqrt{\pi}$. Let $x_t(r)$ denote the intermediate vector $x_t$ at time $t$ with the random seed $r$. For two independent seeds $r, r'$:
\[
    \Pr\!\bigl[\,\exists\, t \in \{1, \ldots, T\} :\; \|x_t(r) - x_t(r')\|_\infty < \delta_q/2 \bigr] \;\leq\; T \cdot \left(\frac{\delta_q}{2\sigma_{\min}\sqrt{\pi}}\right)^{\!d}.
\]
\end{theorem}

\begin{proof}
We bound the probability that the $L^\infty$ distance between two single-step input states falls below $\delta_q/2$, then apply a union bound over $T$ steps.

\paragraph{Step $T$ (initial noise).}
The states $x_T(r)$ and $x_T(r')$ are independent samples from $\mathcal{N}(0, I_d)$. For each coordinate~$j$, $x_{T,j}(r) - x_{T,j}(r') \sim \mathcal{N}(0, 2)$. By Lemma~\ref{lem:anti-conc-app} with $a = 0$ and $\sigma = \sqrt{2}$:
\[
    \Pr\!\bigl[|x_{T,j}(r) - x_{T,j}(r')| \leq \delta_q/2\bigr] \;\leq\; \frac{\delta_q}{2\sqrt{\pi}}.
\]
Since coordinates are independent, the probability that all $d$ coordinates simultaneously lie within $\delta_q/2$ is at most $\bigl(\delta_q/(2\sqrt{\pi})\bigr)^d$.

\paragraph{Steps $t = T{-}1, \ldots, 1$.}
From the DDPM recurrence in Section~\ref{subsec:diffusion-models}, $x_t = f_\theta(x_{t+1}, x_\mathsf{u}, t{+}1) + \sigma_{t+1}\, z_{t+1}$, where $z_{t+1}(r), z_{t+1}(r')$ are independent $\mathcal{N}(0, I_d)$ vectors. For each coordinate~$j$, let $a_j = f_\theta(x_{t+1}(r), x_\mathsf{u}, t{+}1)_j - f_\theta(x_{t+1}(r'), x_\mathsf{u}', t{+}1)_j$. Then
\[
x_{t,j}(r) - x_{t,j}(r') \;=\; a_j
    +\; \sigma_{t+1}\bigl[z_{t+1,j}(r) - z_{t+1,j}(r')\bigr].
\]
Conditioned on $(x_{t+1}(r), x_{t+1}(r'))$, $a_j$ is deterministic and $z_{t+1,j}(r) - z_{t+1,j}(r') \sim \mathcal{N}(0, 2)$. Lemma~\ref{lem:anti-conc-app} with $a = a_j$ and $\sigma = \sigma_{t+1}\sqrt{2}$ gives
\[
    \Pr\!\bigl[|x_{t,j}(r) - x_{t,j}(r')| \leq \delta_q/2 \;\big|\; x_{t+1}(r), x_{t+1}(r')\bigr] \;\leq\; \frac{\delta_q}{2\sigma_{t+1}\sqrt{\pi}}\,.
\]
Independence of the noise coordinates yields a per-step bound of $\bigl(\delta_q/(2\sigma_{t+1}\sqrt{\pi})\bigr)^d$.

\paragraph{Union bound.}
Using $\sigma_{t+1} \geq \sigma_{\min}$ and the initial noise standard deviation $1 \geq \sigma_{\min}$,
\[
    \Pr\bigl[\exists\, t : \|x_t(r) - x_t(r')\|_\infty < \delta_q/2 \bigr] \;\leq\; T \cdot \left(\frac{\delta_q}{2\sigma_{\min}\sqrt{\pi}}\right)^{\!d}. \qedhere
\]
\end{proof}

\paragraph{Numerical instantiation.}
For Stable Diffusion v1-4 ($d = 16{,}384$, $T = 50$, $\sigma_{\min} \approx 0.01$, $\delta_q = 2^{-10} \approx 10^{-3}$), the per-coordinate probability is $\delta_q/(2\sigma_{\min}\sqrt{\pi}) \approx 0.028$, the per-step bound is $0.028^{16{,}384} < 10^{-25{,}000}$, and the union bound over $T = 50$ steps leaves the probability astronomically small.

\subsubsection{No Reuse Inside $f_\theta$}
\label{app:amortizability}

Theorem~\ref{thm:cia-theorem} controls the input to $f_\theta$ at every step. To rule out reuse \emph{inside} $f_\theta$, we adopt the following assumption.

\begin{assumption}[Computational independence of $f_\theta$ under distinct inputs]
\label{def:ftheta-assumption}
Let $f_\theta : \mathbb{R}^d \times \mathbb{R}^{d_\mathsf{u}} \times [T] \to \mathbb{R}^d$ be the denoising network with $L$ layers and layer-$\ell$ intermediate activations $h_\ell \in \mathbb{R}^{d_\ell}$ ($d_\ell$ being the dimension of the $\ell$-th layer). For two inputs $(x_t, x_\mathsf{u}, t)$ and $(x_t', x_\mathsf{u}', t')$ to $f_\theta$ where $x_t - x_t' \sim \mathcal{N}(0,\, 2\sigma_{\min}^2 I_d)$ (i.e. assuming they differ only by the gaussian noise induced by the diffusion model),
\[
    \Pr\!\bigl[\exists\, \ell \in [L] : \|h_\ell(x_t,x_\mathsf{u},t) - h_\ell(x_t',x_\mathsf{u}',t')\|_\infty < \delta_q/2 \bigr] = \mathsf{negl}(d).
\]
We additionally assume that the cost of evaluating $f_\theta(x_t,x_\mathsf{u},t)$ given oracle access to the full computation trace of any prior evaluation of~$f_\theta$ is at least $(1 - \mathsf{negl}(d)) \cdot \mathsf{Depth}(f_\theta)$.
\end{assumption}

\paragraph{Geometric intuition.}
The denoising model's updates guide it towards regions where the target image manifold lies. Figure~\ref{fig:manifold-diagram} illustrates this: when two input positions are different, $f_\theta$ needs to point in different directions so that they advance towards the target manifold.. 

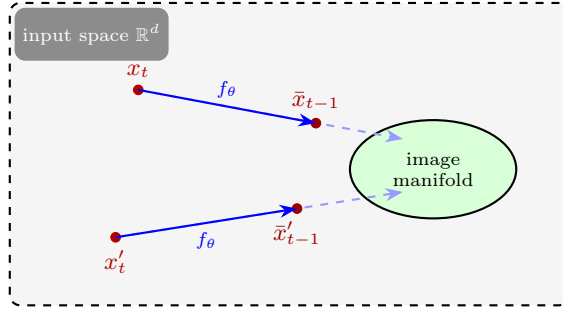
\begin{figure}[t]
\centering
\begin{tikzpicture}[>=Stealth, font=\small]
  \draw[dashed, gray, fill=gray!8] (-3.6,-2) rectangle (3.8,2);
  \node[gray, anchor=north west] at (-3.55, 1.95) {\scriptsize input space $\mathbb{R}^d$};
  \draw[thick, fill=green!15] (2.0, -0.2) ellipse (1.1 and 0.65);
  \node[align=center, font=\scriptsize] at (2.0, -0.2) {image\\manifold};
  \draw[dashed, blue!40, thick, ->] (-1.9, 0.85) -- (1.60, 0.20);
  \draw[dashed, blue!40, thick, ->] (-2.2,-1.10) -- (1.60,-0.50);
  \filldraw[red!70!black] (-1.9, 0.85) circle (1.8pt);
  \node[anchor=south, red!70!black] at (-1.9, 0.90) {\small $x_t$};
  \filldraw[red!70!black] (-2.2,-1.10) circle (1.8pt);
  \node[anchor=north, red!70!black] at (-2.2,-1.15) {\small $x_t'$};
  \filldraw[red!60!black] (0.45, 0.41) circle (1.8pt);
  \node[anchor=south, red!60!black] at (0.45, 0.46) {\small $\bar{x}_{t-1}$};
  \filldraw[red!60!black] (0.20,-0.72) circle (1.8pt);
  \node[anchor=north, red!60!black] at (0.20,-0.77) {\small $\bar{x}_{t-1}'$};
  \draw[thick, ->, blue] (-1.9, 0.85) -- (0.45, 0.41)
      node[midway, above, font=\scriptsize] {$f_\theta$};
  \draw[thick, ->, blue] (-2.2,-1.10) -- (0.20,-0.72)
      node[midway, below, font=\scriptsize] {$f_\theta$};
\end{tikzpicture}
\caption{Geometric intuition for denoising from distinct states. Solid arrows indicate the reverse-step means $\bar{x}_{t-1}=f_\theta(x_t,x_\mathsf{u},t)$ before fresh noise is added; dashed arrows illustrate the overall directions towards the image manifold. }
\label{fig:manifold-diagram}
\end{figure}

\subsection{Empirical Activation Differences}
\label{app:cia-experiments}

We examine evidence for activation separation on Stable Diffusion v1-4~\cite{rombach2022ldm} ($\sim$860M parameters, $d = 16{,}384$, $T = 50$), using a workstation with a 32-core Intel Core~i9 CPU and an NVIDIA RTX~A6000 GPU.

\paragraph{Setup.}
We simulate an adversary attempting to reuse computation from a prior inference on a similar query. We fix a text prompt~$p$ and a base seed~$r$ producing $x_T \sim \mathcal{N}(0, I_d)$, then construct perturbed inputs $x_T + \eta$ with $\eta \sim \mathcal{N}(0, \sigma^2 I_d)$ for $\sigma \in \{0.001, 0.01, 0.1, 0.5, 1.0\}$, representing starting points ranging from near-identical to moderately different. We also include a fully independent-seed baseline. For $N = 1{,}000$ base inputs we record every intermediate noise vector $x_t$ at every denoising step $t$ and report (i)~per-step cosine similarity between base and perturbed vectors, and (ii)~the minimum $L^\infty$ distance across all samples and steps.

Figure~\ref{fig:cosine_sim_app} shows the per-step cosine similarity between base and perturbed vectors. Even for the smallest perturbation, trajectories show notable differences over time. This confirms that $f_\theta$ steers different noise vectors towards different directions, causing denoising paths to differ as predicted.

\begin{figure}[t]
\centering
\includegraphics[width=\columnwidth]{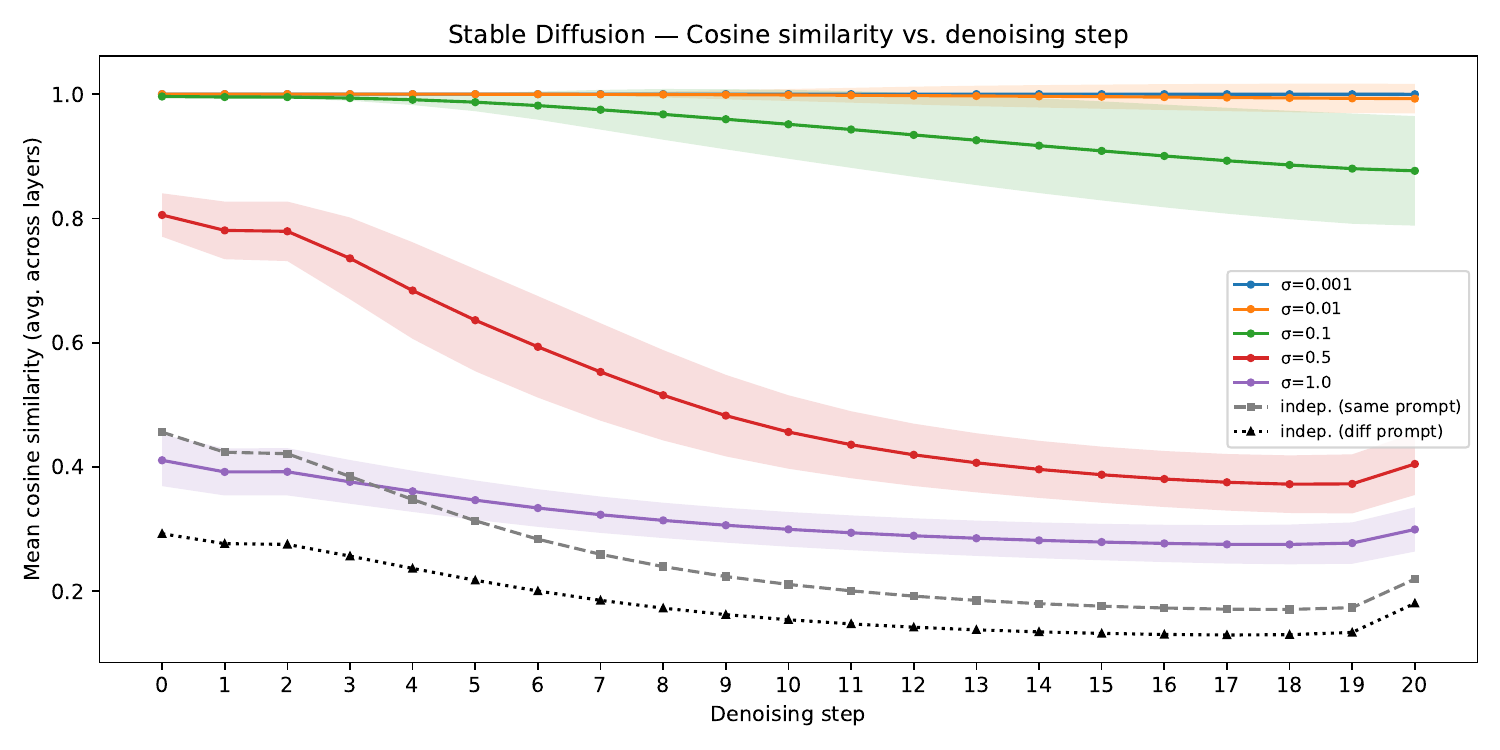}
\caption{Cosine similarity between base and perturbed activations in Stable Diffusion across denoising steps. Even for small perturbations, trajectories differ increasingly. Shaded bands: $\pm 1$ std.\ dev.\ over $N = 1{,}000$ samples.}
\label{fig:cosine_sim_app}
\end{figure}

Table~\ref{tab:min-linf-app} reports the minimum $L^\infty$ distance between base and perturbed noise vectors across all $N = 1{,}000$ samples and all denoising steps---the best case for an adversary attempting reuse. Reuse requires the minimum distance to fall below $\delta_q/2$. For the smallest perturbation ($\sigma = 0.001$), the minimum distance is $2.93 \times 10^{-3}$, exceeding $\delta_q/2$ for all standard ZKML quantization ranges $\delta_q\in [2^{-18}, 2^{-10}]$. In PoML the seeds are fully independent VRF outputs, placing real deployments well beyond the largest perturbation shown.

\begin{table}[t]
  \centering
  \caption{Minimum $L^\infty$ distance between base and perturbed activations across $N = 1{,}000$ samples and all denoising steps, showing a value greater than what most practical deployments require.}
  \label{tab:min-linf-app}
  \begin{tabular}{lr}
    \toprule
    $\boldsymbol{\sigma}$ & \textbf{Min $L^\infty$ distance} \\
    \midrule
    $0.001$ & $0.00293$ \\
    $0.010$ & $0.02832$ \\
    $0.100$ & $0.28076$ \\
    $0.500$ & $0.75488$ \\
    $1.000$ & $2.57617$ \\
    \bottomrule
  \end{tabular}
\end{table}

\subsection{Combined Theorem}
\label{app:final-theorem}

\begin{theorem}[DDPM are $\varepsilon_i$-computationally independent]
\label{thm:amort-diff}
Let $M_\theta$ be a DDPM whose denoising network $f_\theta$ satisfies Assumption~\ref{def:ftheta-assumption}. Then $M_\theta$ satisfies Definition~\ref{def:delta-reuse} for any
\[
    \varepsilon_i \;\geq\; T \cdot \left(\frac{\delta_q}{2\sigma_{\min}\sqrt{\pi}}\right)^{\!d} \;+\; \mathsf{negl}(d).
\]
Since the right-hand side is exponentially small in the latent dimension~$d$, the definition holds for any non-negligible $\varepsilon_i$.
\end{theorem}

\begin{proof}
Consider the computational-independence game of Definition~\ref{def:delta-reuse}. We show the contrapositive: any PPT adversary $\mathcal{A} = (\mathcal{A}_1, \mathcal{A}_2)$ that succeeds with non-negligible probability must incur cost at least $(1 - \mathsf{negl}(d)) \cdot \mathsf{Depth}(M_\theta(x_\mathsf{u}^*, r^*))$ for $\mathcal{A}_2$.

Suppose $\mathcal{A}_2$ outputs the correct $y^* = M_\theta(x_\mathsf{u}^*, r^*)$. By Theorem~\ref{thm:cia-theorem}, for every query made by $\mathcal{A}_1$ with seed $r \neq r^*$, the inputs to $f_\theta$ at each denoising step of the fresh inference differ from those of the evaluation of $M_\theta(x_\mathsf{u}^*, r^*)$ by at least $\delta_q/2$ in $L^\infty$, except with probability $T \cdot (\delta_q/(2\sigma_{\min}\sqrt{\pi}))^d$ -- exponentially small in~$d$. Hence the input-difference condition of Assumption~\ref{def:ftheta-assumption} is satisfied, except with negligible probability in~$d$.

By Assumption~\ref{def:ftheta-assumption}, no intermediate layer activation within $f_\theta$ falls within $\delta_q/2$ of any cached layer activation, so each of the $T$ evaluations of $f_\theta$ required to produce $y^*$ costs at least $(1 - \mathsf{negl}(d)) \cdot \mathsf{Depth}(f_\theta)$. The $(T-1)$ noise additions cannot be skipped either: each uses a distinct noise vector $z_t$ derived from the fresh seed, contributing cost $(T-1) \cdot \Theta(d)$. Therefore, whenever $\mathcal{A}_2$ succeeds:
\begin{align*}
    \mathsf{Depth}(\mathcal{A}_2)
        &\geq T \cdot (1 - \mathsf{negl}(d)) \cdot \mathsf{Depth}(f_\theta) + (T-1)\cdot\Theta(d) \\
        &= (1 - \mathsf{negl}(d)) \cdot \mathsf{Depth}(M_\theta(x_\mathsf{u}^*, r^*)),
\end{align*}
using $\Theta(d) \ll \mathsf{Depth}(f_\theta)$. Equivalently, any $\mathcal{A}_2$ whose cost is below $(1 - \mathsf{negl}(d)) \cdot \mathsf{Depth}(M_\theta(x_\mathsf{u}^*, r^*))$ succeeds with at most negligible probability in~$d$.
\end{proof}

\section{A Reference Complexity Function for GPT-2 and DeepProve}
\label{app:llm-complexity}

We instantiate the combined complexity function of
Section~\ref{subsec:model-specific-protocol} for GPT-2~\cite{radford2019language} and DeepProve~\cite{lagrange_deepprove}. In this definition, we define complexity using operation counts, which depend only on the model and proving
configuration, rather than the device executing them.

\subsection{Model and Proving Configuration}

We use GPT-2 small~\cite{radford2019language}, with 12 Transformer blocks,
hidden width 768, 12 attention heads, feed-forward width 3072, and vocabulary
size 50,257. DeepProve~\cite{lagrange_deepprove} evaluates the quantized model
and proves its execution using sumcheck, LogUp lookup arguments and HyperKZG
polynomial commitments. 

We ran the cryptographic setup phase of the ZKP for a maximum total context of
64 tokens. This setup is reused over multiple proofs and is excluded from the
complexity function. We define the function over two variables $N$ and $K$,
denoting the numbers of input and output tokens, respectively. The function
counts the online inference and ZKP components during the mining.

\subsection{Reference Operations and Complexity Function}
\label{subsec:gpt2-reference-operations}

For the fixed GPT-2 model and DeepProve configuration, we define
\begin{equation}
\boxed{
C_{\theta,\Pi;\mathbf w}(N,K)
=\alpha_{\mathbf w}S+\beta_{\mathbf w}U+\gamma_{\mathbf w}R
 +w_{\rm axis}D(S)+G_{\mathbf w}(P).
}
\label{eq:gpt2-reference-gas}
\end{equation}
The first three terms account for inference and trace reconstruction; the last two account for proof
generation. Each entry of $\mathbf w$ is the weight assigned to one kind of
operation (analogous to how Ethereum assign gas to different operations). This combined reference function is used for fees and lottery weighting. Each component of the complexity function is explained below:

\begin{description}
\item[Linear Transformer operations]
$S=N+K$ counts the total sequence length of input and output tokens. For a fixed transformer, each token
is a vector of the same width and passes through the same layers. The embeddings, linear projections, LayerNorm, residual additions,
activations and output logits therefore require a fixed amount of work per
processed token. Rescaling the newly computed values into an integer
representation for the ZKP, called \emph{requantization}, also contributes per-token work.
The term $\alpha_{\mathbf w}S$ collects these operations. 

\item[Attention work.]
Attention compares each query vector with key vectors from the token
positions it attends to, then combines the corresponding values. Its matrix
multiplications, masking and softmax therefore grow with the number of
query--key pairs. Across the whole decoding pass, this operations needed for the attention mechanism is
\begin{equation}
\begin{aligned}
U&=\underbrace{N^2}_{\text{prefill}}
 +\underbrace{\sum_{j=1}^{K}(N+j)}_{\text{decoding}}
 +\underbrace{S^2}_{\text{trace reconstruction}}\\
 &=N^2+NK+\frac{K(K+1)}2+S^2.
\end{aligned}
\label{eq:gpt2-attention-work}
\end{equation}
The first term processes the prompt, the second generates the output tokens,
and the third rebuilds the complete execution trace for the proof. 

\item[Cache Building]
The key/value cache stores vectors for earlier tokens so decoding can reuse
them. In our implementation, each decoding pass appends the old cache with the
new vectors and requantizes the combined array. Consequently, the
requantization step revisits $N$ cached positions on the first pass, $N+1$ on
the second, and so on. The total number of cache operations is
\[
R=\sum_{j=1}^{K}(N+j-1)=NK+\frac{K(K-1)}2.
\]

\item[Proof reductions.]
When proving the model's matrix operations, DeepProve forms weighted sums
along rows, columns or other dimensions of the arrays involved. These sums
reduce large arrays to smaller quantities used in the proof. This operation is called multiply--accumulate, and we count the number of these operations as
\begin{equation}
D(S)=123863808+111360S+144S^2.
\label{eq:gpt2-axis-work}
\end{equation}
We obtained this expression by enumerating the reductions for GPT-2's layers including
query/key/value projections, attention products, 
feed-forward layers and final vocabulary decoding. 

\item[ZKP Work.]
For many remaining steps within the ZKP, DeepProve arranges its numerical data in
arrays whose sequence dimension is a power of two, adding zero entries when
needed. Hence we define
\[
P=2^{\lceil\log_2 S\rceil}.
\]
For example, a sequence of length 33 uses $P=64$, while a sequence of length
32 uses $P=32$. The work required for the ZKP scales directly with $P$. The operations include
\emph{arithmetic operations};
\emph{lookup operations} that checks if computed value matches some predefined tables;
\emph{commitment operations} which cryptographically bind the prover to its data;
and \emph{challenge-related operations} which record proof transcripts and derive
unpredictable challenges from them. Their combined work effort is
$G_{\mathbf w}(P)=\mathbf w^{\mathsf T}\mathbf B_P$, where $\mathbf B_P$ is
obtained by counting the operations with a fixed $P$.
\end{description}

Table~\ref{tab:gpt2-operation-summary} lists the operations contributing to
these terms. Each operation has its own dedicated weight, which is combined linearly to the coefficient that each operation scales with: $\alpha_{\mathbf w}$ is the weighted
sum of the table's $S$ coefficients, $\beta_{\mathbf w}$ of its $U$
coefficients, and $\gamma_{\mathbf w}$ of its $R$ coefficients.
For example, the three attention operations give
\[
\beta_{\mathbf w}=18432w_{\rm MAC}
 +144w_{\rm mask}+144w_{\rm softmax}.
\]

\begin{table}[tbp]
\centering
\caption{Counted operations and their dependence on sequence lengths.
Numerical inference coefficients and $D(S)$ are derived counts.
For each proof group $g$, $\mathbf B_g(P)$ contains the measured reference
counts at the five padded lengths. Gas weights remain unspecified.}
\label{tab:gpt2-operation-summary}
\footnotesize
\setlength{\tabcolsep}{3pt}
\renewcommand{\arraystretch}{1.15}
\begin{tabularx}{\columnwidth}{@{}l>{\raggedright\arraybackslash}X c >{\raggedright\arraybackslash}p{0.35\columnwidth}@{}}
\toprule
Stage & Counted operations & Depend on & Count or measured profile\\
\midrule
\multirow{10}{*}{Inference}
 & Integer multiply--accumulates & $S,U$ & $247064064S+18432U$\\
 & Token embeddings & $S$ & $1536S$\\
 & Positional embeddings & $S$ & $1536S$\\
 & LayerNorm elements & $S$ & $38400S$\\
 & Element additions & $S$ & $36864S$\\
 & Activation elements & $S$ & $73728S$\\
 & Logit elements & $S$ & $100514S$\\
 & Attention-mask elements & $U$ & $144U$\\
 & Softmax elements & $U$ & $144U$\\
 & Requantization elements & $S,R$ & $150528S+18432R$\\
\midrule
\multirow{5}{*}{ZKP}
 & Arithmetic on actual arrays & $D(S)$ & $D(S)$ in Eq.~\eqref{eq:gpt2-axis-work}\\
 & Arithmetic on padded data & $P$ & $\mathbf B_{\rm arith}(P)$\\
 & Lookup operations & $P$ & $\mathbf B_{\rm lookup}(P)$\\
 & Commitment operations & $P$ & $\mathbf B_{\rm commit}(P)$\\
 & Challenge-related operations & $P$ & $\mathbf B_{\rm challenge}(P)$\\
\bottomrule
\end{tabularx}
\end{table}

We fix the matrix $\mathbf B_P$ by profiling the operations done by DeepProve with different prompt pairs that yields different $(N,K)$ values. The exact definition of this matrix can be found on our GitHub.

\subsection{Experimental Validation}

We tested Eq.~\eqref{eq:gpt2-reference-gas} using synthetic prompts whose
tokens were sampled uniformly from the GPT-2 vocabulary. We covered 40 distinct
$(N,K)$ pairs and 24 total $N+K$ lengths, up to a 64-token context limit. These
included short and longer prompts, short contexts, varied input/output
splits, and lengths around power-of-two boundaries. 

For each execution, we recorded the operation counts during inference and
proof generation and compared them with the counts predicted by the fixed
equation. Let $W_j$ and $\widehat W_j$ denote the recorded and predicted
counts for operation type $j$. We measure the largest relative discrepancy
across all counted inference and proof operations:
\begin{equation}
\epsilon_{\rm trial}
=\max_j\frac{|\widehat W_j-W_j|}{W_j}.
\label{eq:gpt2-count-error}
\end{equation}

Our results show the inference portion of the complexity function
matched exactly. There are however discrepancies in the number of proof operations, as some of them depend on the exact values of the vectors. Across all sample, the maximum error was 3.1088\%,
and the mean of the per-execution maxium error was 0.1468\%. All These
results support the equation as an approximation of the counted work for
the tested  GPT-2/DeepProve configuration. Note that this function will need to be adjusted based and scaled based on the specific PoML deployment.

\section{Formal Treatment of the Backbone Reduction}
\label{app:reduction}

This appendix gives the detailed arguments for the three steps in Section~\ref{sec:reduction}: bounding total work, showing that each completed inference--proof pair yields one valid hash input, and representing the PoML lotteries by independent Bernoulli trials. Work is counted using the integer-valued complexity $C_{\theta,\Pi}$.

\subsection{The Reduction}
\label{subsec:reduction-idea}

We show that PoML inherits the Bitcoin-backbone guarantees of Garay, Kiayias, and Leonardos~\cite{GKL20}---\emph{common prefix}, \emph{chain quality}, and \emph{chain growth}---by satisfying the per-round success bounds used in their analysis.

The Bitcoin backbone analysis~\cite{GKL20} assumes that candidate blocks require a random-oracle hash evaluation and that each (honest or corrupted) party issues at most $q$ queries per round.\footnote{$q$ is a flat-model bound assigning every party the same per-round computation budget; heterogeneous miners are modelled as clusters of flat-model parties.} A query succeeds when its output is below public threshold $T$, independently with probability $T/2^\kappa$, where $\kappa$ is the hash output length. The three guarantees then follow from the properties of independent Bernoulli trials.

Our reduction places PoML in the same template: its lottery (hash evaluation) is equivalent to a $q$-bounded random oracle. A lottery requires honest inference of $M_\theta$ using the prescribed randomness, then hashing the resultant tuple. Define a PoML-oracle query as such hash evaluation $H(\cdot)$, successful when $H(\cdot)<D_1$ for a query with complexity 1.

Suppose each flat-model party has per-round budget $q_{\mathsf{ML}}$, so $\sum_i C_i\leq q_{\mathsf{ML}}$ where $C_i$ are the complexity of the queries. The lottery succeeds when its hash is below the complexity-dependent target $D_{C_i}$ defined in Section~\ref{sec:poml-system-def}, which we will show is equivalent to $C_i$ PoML-oracle queries.

To formalize the analogy, we rule out amortizing inference or proof cost across queries (Step~1) and extra lottery inputs from reordering or re-encrypting outputs (Step~2), then show that PoML is equivalent to at most $q_{\mathsf{ML}}$ independent Bernoulli trials of probability $D/2^\kappa$ (Step~3). Specifically, Step~1 bounds total work by $q_{\mathsf{ML}}$ (as $q$ bounds Bitcoin queries), Step~2 shows that each completed pair yields exactly one hash evaluation, and Step~3 represents the resulting PoML-oracle queries as at most $q_{\mathsf{ML}}$ independent Bernoulli trials.

\subsection{Computational Independence of NIZK Proof Systems}
\label{app:nizk-cia}

Before we present the reduction, since the total work budget includes both inference and generation of a corresponding ZKP, we must also establish a computational independence of the underlying proof system.

\begin{definition}[$\varepsilon_p$-computationally independent NIZK on a distribution $\mathcal{D}$]
\label{def:snark-cia}
Let $\Pi = (\setup, \prover, \verifier)$ be an NIZKPoK system (Definition~\ref{def: NIZK systems}) for an NP-language $L$ with relation $\relation{L}$, let $\mathcal{D} = \{\mathcal{D}_\kappa\}_{\kappa \in \NN}$ be an efficiently sampleable distribution over statement--witness pairs $(x,w) \in \relation{L}$, and let $\varepsilon_p \in [0,1)$. We say $\Pi$ is \emph{$\varepsilon_p$-computationally independent on $\mathcal{D}$} if, for every $\ppt$ adversary $\mathcal{A} = (\mathcal{A}_1,\mathcal{A}_2) = \{(\mathcal{A}_1,\mathcal{A}_2)_\kappa\}_{\kappa \in \NN}$, with $\mathsf{st}$ the state output by $\mathcal{A}_1$, such that the depth of $\mathcal{A}_2$ on each $(x,w)$ is bounded from above by $(1-\varepsilon_p) \cdot \mathsf{Depth}(\prover(\crs,x,w))$, there exists a negligible function $\mu(\cdot)$ such that for all $\kappa \in \NN$:
\begin{equation}
    \Pr\!\left[\, \verifier(\crs, x^*, \pi^*) = 1 \;:\;
    \begin{array}{l}
        \crs \leftarrow \setup(1^\kappa) \\[2pt]
        \mathsf{st} \leftarrow \mathcal{A}_1(\crs) \\[2pt]
        (x^*, w^*) \leftarrow \mathcal{D}_\kappa \\[2pt]
        \pi^* \leftarrow \mathcal{A}_2(\mathsf{st},\, x^*,\, w^*)
    \end{array}
    \right] \;\leq\; \mu(\kappa).
\end{equation}
\end{definition}

\subsubsection{The PoML Distribution $\mathcal{D}_{\text{PoML}}$}

We now specify the distribution under which we require the NIZK in PoML to be computationally independent. Recall that the PoML language $\mathcal{L}_{\text{PoML}}$ (Section~\ref{sec:protocol}) has statements
$
    \mathsf{stmt} =$\\$ (\mathsf{qid},\, h_u,\, \theta,\, \mathcal{R},\, \mathsf{ct},\, \omega,\, \mathsf{pk}_u,\, \mathsf{vk}^{\mathsf{VRF}}_{\mathsf{enc},m},\, r,\, \pi^{\mathsf{VRF}}_{\mathsf{enc}})
$
and witnesses $\mathsf{wit} = (x_\mathsf{u},\, y,\, r_{\mathsf{enc}})$, with $\omega=C_{\theta,\Pi}(x_\mathsf{u},\mathcal{R})$. In an honest execution, every component of $(\mathsf{stmt},\mathsf{wit})$ is determined by the public parameters, the user query, the miner's registered keys, and the seed $r=H(\mathsf{bind}\parallel h_u\parallel\mathsf{qid}\parallel\mathsf{vk}^{\mathsf{sig}}_m)$. This seed determines the inference randomness $\mathcal{R}$ and the encryption randomness $r_\mathsf{enc}$.

\begin{definition}[Distribution $\mathcal{D}_{\text{PoML}}$]
\label{def:poml-distribution}
Fix the public model and complexity function, a query $(\mathsf{pk}_u,\mathsf{qid},h_u,x_\mathsf{u})$ satisfying $h_u=H_{\mathsf{zk}}(x_\mathsf{u}\parallel\mathsf{qid})$, and a miner $m$ with registered key triple $\mathsf{VK}_m$ and corresponding secret keys.
The distribution $\mathcal{D}_{\text{PoML},\kappa}$ over statement--witness pairs of $\mathcal{L}_{\text{PoML}}$ is sampled as follows:
\begin{enumerate}[itemsep=0em,nosep]
    \item Sample $r \leftarrow \{0,1\}^\kappa$ uniformly, modelling the hash output at a fresh mining-seed input.
    \item For $t = 1, \ldots, T(x_\mathsf{u})$, compute $(z_t, \pi^{\mathsf{VRF}}_t) \leftarrow \mathsf{VRF.Eval}(\mathsf{sk}^{\mathsf{VRF}}_{\mathsf{inf},m},\, r \parallel t)$ and set $\mathcal{R} \leftarrow (z_1, \ldots, z_{T(x_\mathsf{u})})$.
    \item Compute $(r_{\mathsf{enc}}, \pi^{\mathsf{VRF}}_{\mathsf{enc}}) \leftarrow \mathsf{VRF.Eval}(\mathsf{sk}^{\mathsf{VRF}}_{\mathsf{enc},m},\, r \parallel \mathsf{qid})$.
    \item Compute $y \leftarrow M_\theta(x_\mathsf{u},\, \mathcal{R})$, $\omega \leftarrow C_{\theta,\Pi}(x_\mathsf{u},\, \mathcal{R})$, and $\mathsf{ct} \leftarrow \mathsf{Enc}(\mathsf{pk}_u,\, y \parallel \mathsf{qid};\, r_{\mathsf{enc}})$.
    \item Output $(\mathsf{stmt}, \mathsf{wit})$ with $\mathsf{stmt} = (\mathsf{qid}, h_u, \theta, \mathcal{R}, \mathsf{ct}, \omega, \mathsf{pk}_u, \mathsf{vk}^{\mathsf{VRF}}_{\mathsf{enc},m}, r, \pi^{\mathsf{VRF}}_{\mathsf{enc}})$ and $\mathsf{wit} = (x_\mathsf{u}, y, r_{\mathsf{enc}})$.
\end{enumerate}
\end{definition}

By construction $(\mathsf{stmt},\mathsf{wit}) \in \relation{\mathcal{L}_{\text{PoML}}}$, and the only source of randomness in the sampler is $r$, which simultaneously seeds the inference VRF (producing $\mathcal{R}$) and the encryption VRF (producing $r_{\mathsf{enc}}$). The induced distribution over statement--witness pairs is the distribution of a single mining inference--proof step at a fresh mining-seed hash input, for the fixed query and registered miner.

\begin{assumption}[Computational independence of Fiat--Shamir zk-SNARKs on $\mathcal{D}_{\text{PoML}}$]
\label{ass:snark-cia-poml}
Let $\Pi = (\setup, \prover, \verifier)$ be a non-interactive zk-SNARK satisfying Def.~\ref{def: NIZK systems} which applies the Fiat--Shamir transform, modelled as a random oracle. Let $M_\theta \in \mathcal{K}$ be a PoML-compatible model. Then $\Pi$ satisfies Definition~\ref{def:snark-cia} on the PoML statement--witness distribution $\mathcal{D}_{\text{PoML}}$ (Definition~\ref{def:poml-distribution}) with $\varepsilon_p = \mathsf{negl}(\kappa)$, where $\kappa$ is the Fiat--Shamir hash output length.
\end{assumption}

\paragraph{Justification.}
We give an informal argument supporting Assumption~\ref{ass:snark-cia-poml}. The Fiat--Shamir transform replaces each verifier message with a hash evaluation over the transcript so far, including the statement~$\mathsf{stmt}$~\cite{CiampiPSSV16,BernhardPW12}; modelling that hash as a random oracle, the challenges are unpredictable functions of~$\mathsf{stmt}$. By the structure of $\mathcal{D}_{\text{PoML}}$ (Definition~\ref{def:poml-distribution}), two independent draws yield, with overwhelming probability, statements that differ in
the randomness collection~$\mathcal{R}$, the seed~$r$, and the ciphertext~$\mathsf{ct}$.

The polynomial commitments---binding functions of the witness coefficients---and the Fiat--Shamir challenges are therefore both fresh and unrelated to any precomputation $\mathcal{A}_1$ may have performed. Saving more than a $\mathsf{negl}(\kappa)$-fraction of the prover's depth would amount to predicting a random-oracle output, contradicting the random-oracle assumption on the Fiat--Shamir hash.

\myparagraph{Remark}: We note that this is also an implicit assumption made implicitly by the paper \cite{pouw} defining PoUW using zk-SNARKs. In their paper, independence property is assumed when the inputs statements differ.

\subsection{Step~1: Total Units of Work Are $q_\mathsf{ML}$-Bounded}
\label{subsec:step1}
\label{app:proof-amort-resist}

The first step replaces the $q$-bound on random-oracle queries with the work spent producing inference outputs. Write $C_i=C_{\theta,\Pi}(x_{\mathsf{u},i},\mathcal{R}_i)$. A pair of complexity $C_i$ requires $C_i$ work units, while a flat-model party has at most $q_{\mathsf{ML}}$ per round.

Definition~\ref{def:model} requires $M_\theta$ to satisfy $\epsilon_i$-computational independence relative to a fresh evaluation of $(x,\mathcal{R})$, forbidding reuse across distinct inferences. Inference therefore costs at least $(1-\epsilon_i)$ of the honest work with no saved states.

We also assume \emph{$\varepsilon_p$-computational independence} for an NIZK proof system over a distribution $\mathcal{D}_\text{PoML}$ of statement--witness pairs: precomputation saves at most an $\varepsilon_p$ fraction of prover computation cost. The formalization is in App.~\ref{app:nizk-cia}.

\begin{theorem}[$q_{\mathsf{ML}}$-bounded work]
\label{thm:amort-resist}
Let $M_\theta\in\mathcal{K}$ be PoML-compatible (Definition~\ref{def:model}) with parameter $\varepsilon_i$. Let the proof system be a non-interactive zk-SNARK, and assume it is $\varepsilon_p$-computationally independent on $\mathcal{D}_{\text{PoML}}$ (Definitions~\ref{def:snark-cia} and~\ref{def:poml-distribution}). If $\max(\varepsilon_i,\varepsilon_p)<1/(q_{\mathsf{ML}}+1)$, then each flat-model party completes total complexity at most $q_{\mathsf{ML}}$ per round, except with negligible probability in $\kappa$.
\end{theorem}

\begin{proof}
Let $T_M$ denote the cost of one unit of work, and let the per-round computation budget of a flat-model party be $T_{\mathsf{round}}=q_{\mathsf{ML}}T_M$, so $q_{\mathsf{ML}}$ is its honest work-unit budget. For an admissible inference--proof pair $i$, write $C_i=C_{\theta,\Pi}(x_i,\mathcal{R}_i)$; the pair has cost $C_iT_M$ (Note that although miners are only required to use the inference result in the lottery, they are essentially committed to generating the proof before the block publishes).

\paragraph{Precomputation resistance.}
We first rule out carrying an exact completed pair across rounds. Observe that all inference randomness is ultimately seeded by $\mathsf{bind}_1 = G(s, \mathsf{tx})$, where $(s, \mathsf{tx})$ encodes the hash of the previous block and the current transaction set. Because $\mathsf{tx}$ is not fixed until the beginning of the round, the adversary cannot determine $\mathsf{bind}_1$, and hence any of $r_1, r_2, \ldots$, before the round starts. Any inference--proof pair computed in a prior round used a different seed $G(s', \mathsf{tx}')$ and will fail the binding check in the current round.

\paragraph{Inference cost.}
By the contrapositive of Definition~\ref{def:delta-reuse}, any PPT adversary that produces a correct inference output $M_\theta(x_\mathsf{u}, \mathcal{R})$ with non-negligible probability must have incurred cost at least $(1-\varepsilon_i) \cdot \mathsf{Depth}(M_\theta(x_\mathsf{u}, \mathcal{R}))$ on that evaluation; otherwise its success probability is negligible. Hence the adversary's inference cost is at least $(1-\varepsilon_i)$ times the reference inference component of its work, except with negligible probability.

\paragraph{Proof generation cost.}
By the contrapositive of Definition~\ref{def:snark-cia}, any $\ppt$ adversary that, on a fresh statement--witness pair $(x^*,w^*) \leftarrow \mathcal{D}_{\text{PoML}}$, outputs a verifying proof $\pi^*$ with non-negligible probability must incur cost at least $(1-\varepsilon_p) \cdot \mathsf{Depth}(\prover(\crs,x^*,w^*))$. Since the honest prover on inputs from $\mathcal{D}_{\text{PoML}}$ supplies the calibrated proof component of the pair's work, this is at least $(1-\varepsilon_p)$ times that component.

\paragraph{Combined bound.}
Set $\varepsilon^\star := \max(\varepsilon_i, \varepsilon_p)$. Our assumption requires completed pairs with total complexity $\sum_{i\in S} C_i$ jointly to cost at least $(1-\varepsilon^\star)\left(\sum_{i\in S} C_i\right)T_M$, yielding at most
\[
    \frac{q_{\mathsf{ML}}}{1-\varepsilon^\star}
\]
work units. Under $\varepsilon^\star < 1/(q_{\mathsf{ML}}+1)$, $q_{\mathsf{ML}}/(1-\varepsilon^\star)<q_{\mathsf{ML}}+1$, so the number of completed integer work units is at most $q_{\mathsf{ML}}$.
\end{proof}

Theorem~\ref{thm:amort-diff} establishes $\varepsilon_i = \mathsf{negl}(d)$ for DDPM (with $d$ the latent dimension), and Assumption~\ref{ass:snark-cia-poml} of Appendix~\ref{app:nizk-cia} posits $\varepsilon_p = \mathsf{negl}(\kappa)$ for any Fiat--Shamir-based zk-SNARK on $\mathcal{D}_{\text{PoML}}$ (with $\kappa$ the hash output length).

\subsection{Step~2: One Query Equates to One Hash Evaluation}
\label{subsec:step2}
\label{app:proof-bijection}

Here one query yields at most one hash evaluation, so an adversary cannot gain additional valid hash inputs.

\begin{lemma}[Uniqueness of Hash Evaluation Inputs]
\label{lem:bijection}
Suppose a party completes $n$ inference--proof pairs of total complexity at most $q_{\mathsf{ML}}$ per round, and all randomness is derived through an unbiasable VRF (\S\ref{def:vrf-app}). Then, each party obtains at most $n$ valid inputs to the hash function $H$.
\end{lemma}

\begin{proof}
\red{Fix a miner with registered key triple $\mathsf{VK}_m=(\mathsf{vk}^{\mathsf{sig}}_m,\mathsf{vk}^{\mathsf{VRF}}_{\mathsf{inf},m},\mathsf{vk}^{\mathsf{VRF}}_{\mathsf{enc},m})$ and block header $(s,\mathsf{tx})$. We first establish two uniqueness properties.}

\paragraph{Inference uniqueness.}
Fix a user query $(x_{\mathsf{u},j}, h_{u,j}, \mathsf{qid}_j)$ and a position index~$i$. The ZK proof attached to each query prevents any forging of $\mathsf{ct}_i$. For an honestly computed inference, the seed
{
\[
    r_i \;=\; H(\mathsf{bind}_i \,\|\, h_{u,j} \,\|\, \mathsf{qid}_j \,\|\, \mathsf{vk}^{\mathsf{sig}}_m)
\]
}
is deterministic in its inputs. Registration and Definition~\ref{def:valid-block} require the same $\mathsf{VK}_m$ at every position. Unique provability under $\mathsf{vk}^{\mathsf{VRF}}_{\mathsf{inf},m}$ therefore forces a single noise schedule $\mathcal{R}_i=(z_{i,1},\ldots,z_{i,T(x_{\mathsf{u},j})})$ for which the published inference transcript can pass $\mathsf{VRF.Vfy}$; any alternative $\mathcal{R}_i'\neq\mathcal{R}_i$ would contradict unique provability. Since $M_\theta$ is deterministic once $(x_\mathsf{u},\mathcal{R})$ is fixed (Section~\ref{subsec:model-class}), $y_i$ is uniquely determined. The encryption VRF under $\mathsf{vk}^{\mathsf{VRF}}_{\mathsf{enc},m}$ likewise fixes $r_{\mathsf{enc},i}$ on input $r_i\parallel\mathsf{qid}_j$; hence determinism of encryption fixes $\mathsf{ct}_i$.

\paragraph{Prefix uniqueness of $\mathcal{Y}_i$ for fixed final element.}
\red{For $i \geq 2$, the binding $\mathsf{bind}_i = H(\cdot\parallel\mathcal{Y}_{i-1})$ ties $\mathsf{ct}_i$ to its predecessor inference chain. By collision resistance of $H$, $\mathsf{bind}_i$ uniquely identifies $(\mathsf{ct}_1, \ldots, \mathsf{ct}_{i-1})$ up to negligible probability. This shows that any given valid $\mathsf{ct}_i$ is bound to its entire predecessor chain $(\mathsf{ct}_1, \ldots, \mathsf{ct}_{i-1})$. Definition~\ref{def:valid-block} additionally rejects any purported prefix that changes $\mathsf{VK}_m$ between positions. Hence a completed $\mathsf{ct}_i$ can be appended to exactly one valid prefix, yielding a single lottery input $\mathcal{Y}_i = (\mathsf{ct}_1, \ldots, \mathsf{ct}_i)$.}

For a set $S$ of $n$ completed pairs, prefix uniqueness gives one lottery input $G(s,\mathsf{tx})\parallel\mathcal{Y}_i$ for each output $i$. Repeated evaluations of the same input do not yield different results. The party therefore obtains at most $n$ valid hash inputs.

\end{proof}

\subsection{Step~3: PoML Lotteries Are $q_\mathsf{ML}$-Bounded Independent Bernoulli Trials}
\label{subsec:step3}
\label{app:proof-lottery-equiv}

Combining Steps 1 and 2, each party submits evaluations $H(G(s,\mathsf{tx})\parallel\mathcal{Y}_i)$, one per inference--proof pair, with total work at most $q_\mathsf{ML}$. This gives the Bitcoin $q$-bounded random-oracle structure.

\begin{lemma}[PoML Lotteries Are Independent Bernoulli Trials]
\label{lem:lottery-equiv}
Under the random-oracle model for $H$, each party gets at most $q_\mathsf{ML}$ queries to the PoML-oracle, where each query to the PoML-oracle is an independent Bernoulli trial with success probability $D/2^\kappa$. Hence, the PoML lottery mechanism is equivalent to at most $q_\mathsf{ML}$ independent Bernoulli($D/2^\kappa$) trials per round.
\end{lemma}

\begin{proof}
Because $H$ is a random oracle and the inputs are distinct, these evaluations are uniform and independent. Let $p=D/2^\kappa=\tau_D(1)$. A PoML-oracle query corresponds to a Bernoulli$(p)$ trial, the lottery of an inference--proof pair of complexity $1$. We first show that one completed inference-proof pair of complexity $C_i$ is equivalent to $C_i$ oracle queries.

Note that for $C_i$ independent Bernoulli$(p)$ trials $B_{i,1},\ldots,B_{i,C_i}$,
\[
  \Pr\!\left[\bigcup_{j=1}^{C_i}\{B_{i,j}=1\}\right]
  =1-(1-p)^{C_i}=\tau_D(C_i).
\]

Thus, each PoML hash evaluation with complexity $C_i$ has the same distribution as $\mathbf{1}\{\sum_{j=1}^{C_i}B_{i,j}>0\}$, equivalent to running the PoML oracle $C_i$ times, up to negligible error $2^{-\kappa}$. Since total complexity is at most $q_\mathsf{ML}$, the PoML lottery is therefore equivalent to at most $q_\mathsf{ML}$ independent Bernoulli($D/2^\kappa$) trials per round.
\end{proof}

\subsection{Typicality of PoML Executions}
\label{app:reduction-typicality}

The preceding steps establish the per-round Bernoulli representation. We now verify that this representation satisfies the typicality conditions used in the backbone security proof~\cite{GKL20}. Let $X_r$ indicate that at least one honest party succeeds in round $r$, and let $Y_r$ indicate that exactly one honest party succeeds. Let $Z_r$ count successful adversarial Bernoulli trials. Under the honest-work condition used below---each honest party either wins or completes $q_{\mathsf{ML}}$ units of work in each round---the backbone bounds give
\begin{align*}
  \mathbb{E}[X_r]&=f=1-(1-p)^{(n-t)q_{\mathsf{ML}}},\\
  \mathbb{E}[Y_r]&\geq f(1-f),\\
  \mathbb{E}[Z_r]&=tq_{\mathsf{ML}}p\leq\frac{t}{n-t}\frac{f}{1-f}.
\end{align*}
Write $X(S)=\sum_{r\in S}X_r$ and analogously for $Y(S)$ and $Z(S)$. The following typicality event is defined on this representation.

\begin{definition}[$(\varepsilon,\lambda)$-typical execution~{\cite[Def.~9]{GKL20}}]
\label{def:typical-app}
An execution is \emph{$(\varepsilon,\lambda)$-typical} for $\varepsilon \in (0,1)$ and integer $\lambda \geq 2/f$ if, for every set $S$ of at least $\lambda$ consecutive rounds:
\begin{enumerate}[label=(\alph*),itemsep=0.2em]
    \item $(1-\varepsilon)\mathbb{E}[X(S)] < X(S) < (1+\varepsilon)\mathbb{E}[X(S)]$ \;\text{and}\; $(1-\varepsilon)\mathbb{E}[Y(S)] < Y(S)$;
    \item $Z(S) < \mathbb{E}[Z(S)] + \varepsilon\,\mathbb{E}[X(S)]$;
    \item no \emph{insertions}, \emph{copies}, or \emph{predictions} occurred.
\end{enumerate}
\end{definition}

By Theorem~10 of~\cite{GKL20}, an execution is $(\varepsilon,\lambda)$-typical except with probability $e^{-\Omega(\varepsilon^2 f \lambda + \kappa - \log L)}$.

The resulting concentration statement for PoML is as follows.

\begin{lemma}[Typicality under PoML]
\label{lem:typicality}
PoML executions are $(\varepsilon, \lambda)$-typical with probability at least $1 - e^{-\Omega(\varepsilon^2 f \lambda \,+\, \kappa \,-\, \log L)}$.
\end{lemma}

\begin{proof}
By Lemma~\ref{lem:lottery-equiv}, each party performs lottery evaluations per round equivalent to at most $q_\mathsf{ML}$ independent Bernoulli trials with success probability $D/2^\kappa$. This is exactly the distributional structure assumed of the backbone's per-round oracle, with parameters $q = q_{\mathsf{ML}}$ and $T = D$. Therefore the random variables $X_r, Y_r, Z_r$ in PoML satisfy the same independence and distributional properties as in the backbone, and the proof of Theorem~10 in~\cite{GKL20}---Chernoff bounds for parts~(a) and~(b) of Definition~\ref{def:typical-app}, and standard collision/prediction bounds for part~(c)---applies verbatim with $q = q_{\mathsf{ML}}$ and $T = D$.
\end{proof}

The next lemma collects the backbone bounds needed below. We reproduce them
from~\cite{GKL20} for self-containment, with the Bernoulli trials represented
by baseline parameters $q = q_{\mathsf{ML}}$, $T = D$, and $\delta = \gamma$.

\begin{lemma}[\cite{GKL20}, Lemma~11, restated]
\label{lem:gkl-lemma11-app}
For any set $S$ of at least $\lambda$ consecutive rounds in a typical execution:
\begin{enumerate}
    \item[(a)] $(1 - \varepsilon)f|S| < X(S) < (1 + \varepsilon)f|S|$ \;and\; $\bigl(1 - \tfrac{\gamma}{3}\bigr)f|S| < (1 - \varepsilon)f(1 - f)|S| < Y(S)$.
    \item[(b)] $Z(S) < \tfrac{t}{n-t}\cdot\tfrac{f}{1-f}\cdot|S| + \varepsilon f|S| < \bigl(1 - \tfrac{2\gamma}{3}\bigr)f|S|$.
    \item[(c)] $Z(S) < \bigl(1 + \tfrac{\gamma}{2}\bigr)\cdot\tfrac{t}{n-t}\cdot X(S) + \varepsilon f|S| < (1 - \tfrac{\gamma}{2})X(S)$ \;and\; $Z(S) < Y(S)$.
\end{enumerate}
\end{lemma}

\subsection{Main Theorem: PoML Guarantees}
\label{subsec:main-thm}

Combining the work bound, the uniqueness lemma, the Bernoulli representation,
and the typicality result above gives the main reduction. Here $n$ is the number
of miners, $t$ the number corrupted, $f$ the probability that an honest party
wins in a round, $\gamma$ the honest-majority margin defined by
$t/(n-t)\leq1-\gamma$, and $\varepsilon\in(0,1)$ a concentration parameter
bounding deviation from the expected number of successful honest rounds.

\begin{theorem}[PoML inherits backbone guarantees]
\label{thm:main-reduction}
Consider the PoML protocol (Alg.~\ref{alg:poml-main}) with $n$ flat-model parties, at most $t$ corrupted, security parameter $\kappa$, and difficulty $D$. Suppose:
\begin{itemize}[itemsep=0em,nosep]
    \item the model is PoML-compatible (Definition~\ref{def:model}) and the proof system is a non-interactive zk-SNARK that is $\varepsilon_p$-computationally independent on $\mathcal{D}_{\text{PoML}}$ (Definitions~\ref{def:snark-cia} and~\ref{def:poml-distribution}), with $\max(\varepsilon_i,\varepsilon_p)<1/(q_{\mathsf{ML}}+1)$;
    \item a corresponding, easily computable complexity function $C_{\theta,\Pi}$ accurately reflects the unit of work for each input;
    \item there is always enough queries for miners to work on (thus filling their $q_\mathsf{ML}$ work budget per round);
    \item difficulty $D$ gives per-round success probability $f$ satisfying $t\leq(1-\gamma)(n-t)$ and $3f+3\varepsilon<\gamma$ for some $\varepsilon,\gamma\in(0,1)$.
\end{itemize}
Then, except with probability negligible in $\kappa$, the PoML protocol satisfies the \emph{common prefix}, \emph{chain quality} and \emph{chain growth} properties of~\cite{GKL20}, with the same parameter formulas as for Bitcoin under the substitutions $q \mapsto q_{\mathsf{ML}}$ and $T \mapsto D$. The corresponding explicit parameter instantiation follows in the next subsection.
\end{theorem}

\begin{proof}
By Theorem~\ref{thm:amort-resist} and Lemma~\ref{lem:bijection}, each party obtains at most $q_{\mathsf{ML}}$ virtual tickets per round. Lemma~\ref{lem:lottery-equiv} represents these tickets as independent Bernoulli($D/2^\kappa$) trials. The typicality lemma and the backbone bounds above therefore apply with $q\leftarrow q_{\mathsf{ML}}$ and $T\leftarrow D$; the proofs of Theorems~12, 15, and~16 of~\cite{GKL20} then yield the stated guarantees.
\end{proof}

\paragraph{Parameter calibration.}
Because ML inference is slower than hashing, $q_{\mathsf{ML}}\ll q$ and equal difficulty would make PoML's $f$ much smaller than PoW's. PoML therefore adjusts $f$ through $D$. For target average block time $\tau_b$ and round duration $\Delta$, let $f=\Delta/\tau_b$. Ignoring integer rounding, the $(n-t)$ honest parties evaluate $(n-t)q_{\mathsf{ML}}$ virtual tickets per round, giving a difficulty of $D=2^\kappa\!\left(1-(1-f)^{1/((n-t)q_{\mathsf{ML}})}\right)$, which is approximately $f2^\kappa/((n-t)q_{\mathsf{ML}})$ when $f$ is small.

\subsection{Explicit Parameters for PoML Guarantees}
\label{app:reduction-main-thm}

We now instantiate the backbone bounds for Theorem~\ref{thm:main-reduction}.
The substitutions are $q\mapsto q_{\mathsf{ML}}$ and $T\mapsto D$, as established
by the Bernoulli-trial representation and the typicality lemma above.

\begin{theorem}[Explicit backbone parameters]
\label{thm:main-reduction-app}
Under the assumptions of Theorem~\ref{thm:main-reduction}, the following hold for any $(\varepsilon, \lambda)$-typical execution (whose failure probability, including threshold rounding, is bounded in Lemma~\ref{lem:typicality}):
\begin{enumerate}[itemsep=0em,nosep]
    \item \textbf{Common Prefix} (\cite{GKL20}, Theorem~15) with parameter $k \geq 2\lambda f$: for any two honest parties with chains $\mathcal{C}_1, \mathcal{C}_2$ adopted at rounds $r_1 \leq r_2$, $\mathcal{C}_1^{\lceil k} \preceq \mathcal{C}_2$.
    \item \textbf{Chain Quality} (\cite{GKL20}, Theorem~16) with parameters $\mu$ and $\ell \geq 2\lambda f$: in any $\ell$ consecutive blocks of any honest party's chain, the fraction of honest blocks is at least $\mu$, where $\mu = 1 - \bigl(1 + \tfrac{\gamma}{2}\bigr)\cdot \tfrac{t}{n-t} - \tfrac{\varepsilon}{1-\varepsilon}.$
    \item \textbf{Chain Growth} (\cite{GKL20}, Theorem~12) with parameters $\tau$ and $s \geq \lambda$: for any $s$ consecutive rounds, the chain of any honest party grows by at least $\tau \cdot s$ blocks, with $\tau = (1 - \varepsilon)f$.
\end{enumerate}
\end{theorem}

\begin{proof}
Theorem~\ref{thm:amort-resist} and Lemmas~\ref{lem:bijection} and~\ref{lem:lottery-equiv} represent each party's lotteries using at most $q_{\mathsf{ML}}$ independent Bernoulli trials per round. By Lemma~\ref{lem:typicality}, PoML executions are $(\epsilon,\lambda)-$typical with the stated probability. Therefore, Lemma~\ref{lem:gkl-lemma11-app} applies in PoML and Theorems~12, 15 and~16 of~\cite{GKL20} follow directly.
\end{proof}

\section{Additional Wasted-Work Results}
\label{app:wasted-work}

The corresponding heatmaps for all
three target block times are shown here. Each panel reports the mean percentage
of completed wasted complexity over the same $(M,Q)$ configurations described in Section~\ref{sec:useful-work-experiment}.

\begin{figure*}[p]
  \centering
  \begin{subfigure}{0.68\linewidth}
    \centering
    \includegraphics[width=\linewidth]{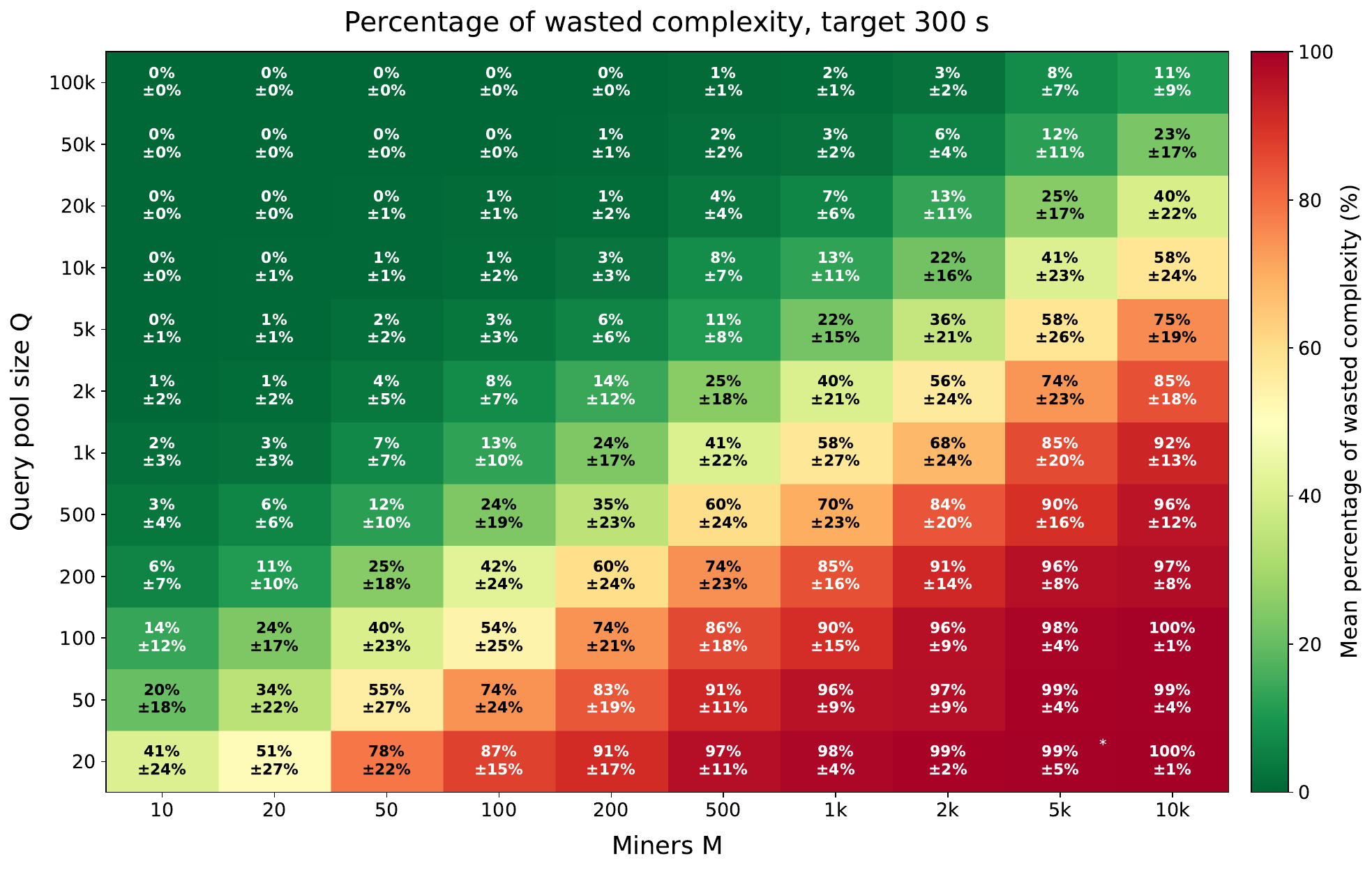}
    \caption{Mean percentage of completed wasted complexity at a target block time of $300$~s.}
    \label{fig:uniform-fee-waste-300}
    \label{fig:uniform-fee-waste}
  \end{subfigure}

  \medskip
  \begin{subfigure}{0.68\linewidth}
    \centering
    \includegraphics[width=\linewidth]{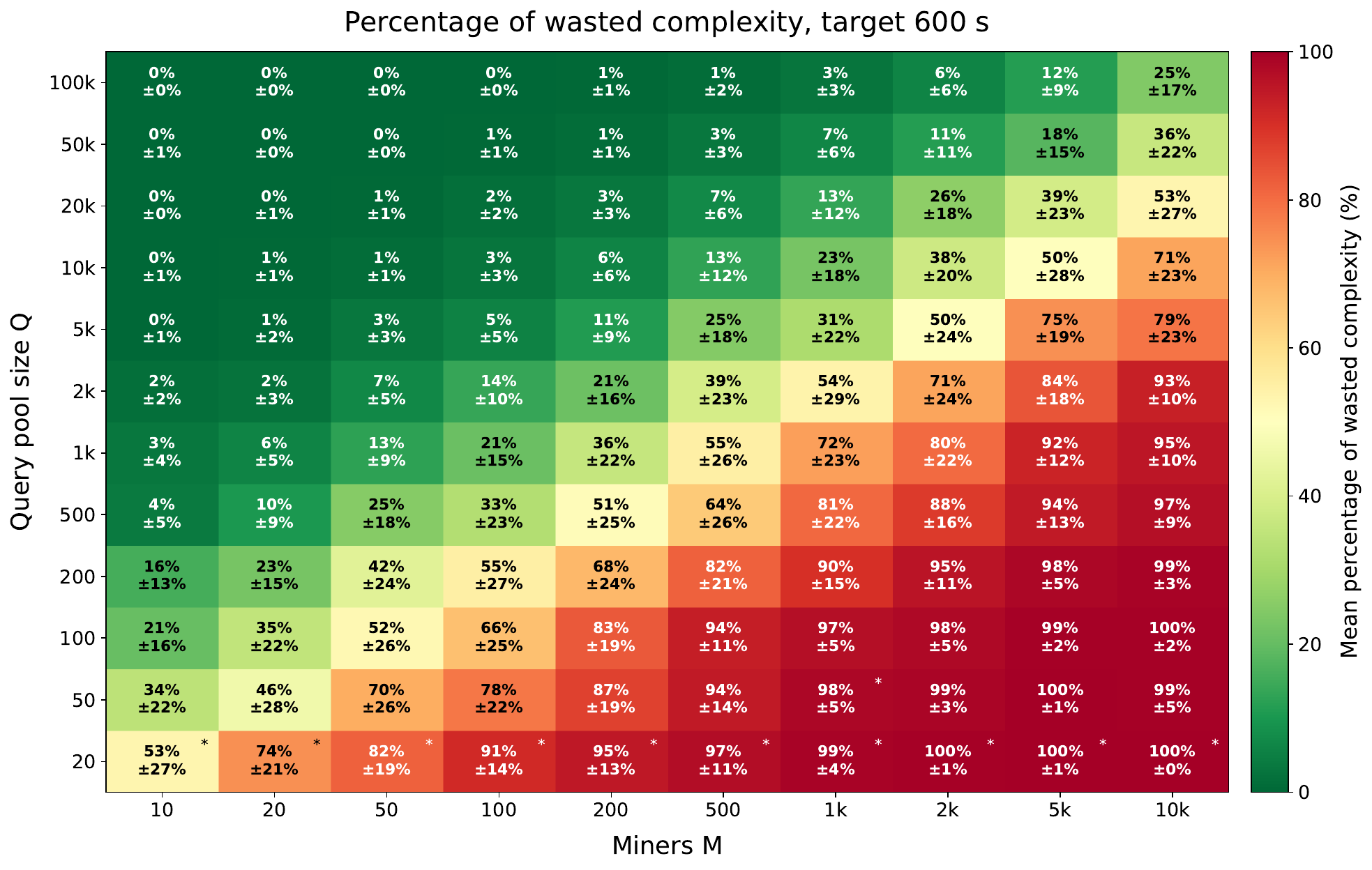}
    \caption{Mean percentage of completed wasted complexity at a target block time of $600$~s.}
    \label{fig:uniform-fee-waste-600}
  \end{subfigure}

  \medskip
  \begin{subfigure}{0.68\linewidth}
    \centering
    \includegraphics[width=\linewidth]{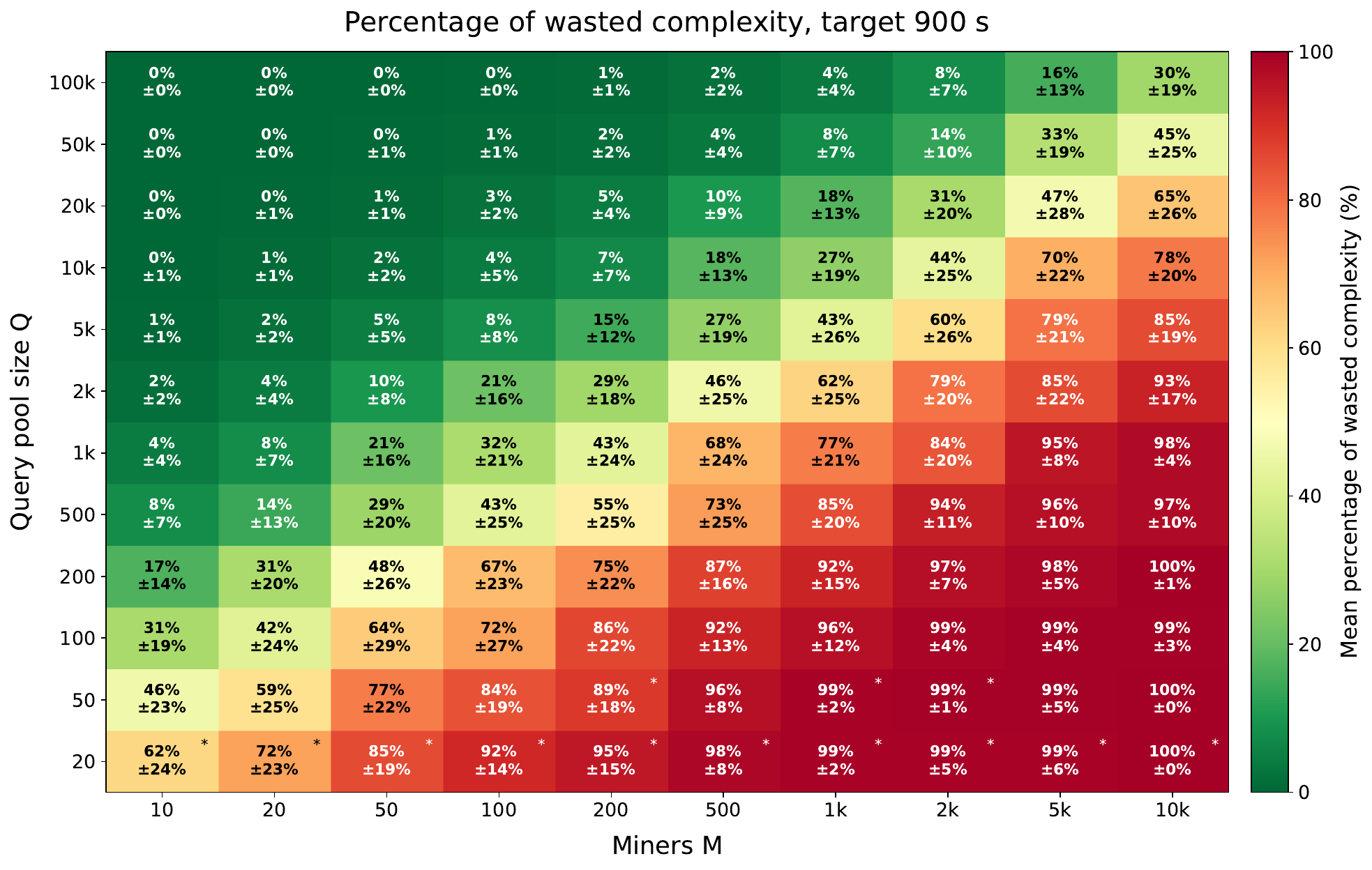}
    \caption{Mean percentage of completed wasted complexity at a target block time of $900$~s.}
    \label{fig:uniform-fee-waste-900}
  \end{subfigure}
\end{figure*}
\clearpage

\section{Block-Time Calculation}
\label{app:scalability}

This appendix provides the full derivation of the block-time cited
in Section~\ref{sec:scalability-discussion}.

We use as reference the zkLLM~\cite{sun2024zkllm} prover benchmark: an inference--proof pair for a $125$M-parameter model completes in $73.9$~s on a single A100 GPU. We match the round time $\Delta = 73.9$~s with one pair per round so $q_{\mathsf{ML}} = 1$.
We fix $\varepsilon = 0.05$ (standard in Chernoff-bound).

The per-round success probability $f$ governs the security--liveness trade-off in the Bitcoin Backbone framework.
$f$ must be small enough that the honest-majority condition (Equation~\ref{eq:honest-majority} below) is satisfied with a meaningful margin, yet large enough that block production remains practical (Equation~\ref{eq:block-time}).
With $\Delta = 73.9$~s, setting $f = 0.10$ yields a block time of approximately $12.3$ minutes, roughly comparable to Bitcoin's $10$-minute block target, with an adversarial-tolerance bound of $\sim 34.6\%$. The derivation is as follows.

\paragraph{Block time.}
Using the difficulty relation in Section~\ref{subsec:main-thm}, the expected block time satisfies
\begin{equation}
    \tau_b \;=\; \frac{\Delta}{f}.
    \label{eq:block-time}
\end{equation}
Substituting $\Delta = 73.9$~s and $f = 0.10$:
\begin{equation*}
    \tau_b \;=\; \frac{73.9}{0.10} \;=\; 739 \text{ s} \;\approx\; 12.3 \text{ minutes}.
\end{equation*}

\paragraph{Honest-majority margin.}
Theorem~\ref{thm:main-reduction} requires the honest-majority condition
\begin{equation}
    3f + 3\varepsilon \;<\; \gamma,
    \label{eq:honest-majority}
\end{equation}
where $\gamma \in (0, 1]$ is defined via $t / (n - t) \le 1 - \gamma$
\cite[\S 4.1]{GKL20}. Substituting $f = 0.10$ and
$\varepsilon = 0.05$, and adding a $0.02$ safety margin so that the
inequality is strict in practice:
\begin{equation*}
    \gamma \;>\; 3(0.10) + 3(0.05) + 0.02 \;=\; 0.47.
\end{equation*}
Choosing $\gamma = 0.47$ at the lower bound, the honest-majority condition
$t \le (1 - \gamma)(n - t)$ rearranges to
\begin{equation}
    \frac{t}{n} \;\le\; \frac{1 - \gamma}{2 - \gamma},
    \label{eq:adv-fraction}
\end{equation}
yielding
\begin{equation*}
    \frac{t}{n} \;\le\; \frac{1 - 0.47}{2 - 0.47} \;=\; \frac{0.53}{1.53} \;\approx\; 34.6\%.
\end{equation*}

\end{document}